\documentclass[letterpaper, 10 pt, conference]{ieeeconf}  % Comment this line out if you need a4paper

\IEEEoverridecommandlockouts                              % This command is only needed if 
\usepackage{cite}

\usepackage{amsmath,amssymb,amsfonts}
\usepackage{amsthm}
\usepackage{graphicx}
\usepackage{textcomp}
\usepackage{hyperref}
\usepackage{xcolor}
\usepackage{algpseudocode}
\usepackage{tikz}
\usepackage{algorithm}
\usepackage{subfigure}

\DeclareMathOperator*{\argmin}{arg\,min}
\newtheorem{remark}{Remark}
\newtheorem{assumption}{Assumption}
\newtheorem{definition}{Definition}
\newtheorem{proposition}{Proposition}
\newtheorem{theorem}{Theorem}
\newtheorem{lemma}{Lemma}
\newtheorem{corollary}{Corollary}

\title{\LARGE \bf
Routing Equilibrium in Mixed-Autonomy Traffic Networks with Altruistic Autonomous Agents
}

\author{Albert Author$^{1}$ and Bernard D. Researcher$^{2}$% <-this % stops a space
\thanks{*This work was not supported by any organization}% <-this % stops a space
\thanks{$^{1}$Albert Author is with Faculty of Electrical Engineering, Mathematics and Computer Science,
        University of Twente, 7500 AE Enschede, The Netherlands
        {\tt\small albert.author@papercept.net}}%
\thanks{$^{2}$Bernard D. Researcheris with the Department of Electrical Engineering, Wright State University,
        Dayton, OH 45435, USA
        {\tt\small b.d.researcher@ieee.org}}%
}

\author{Lihui Yi$^{1}$ and Ermin Wei$^{1,2}$% <-this % stops a space
\thanks{{This work was supported in part by the National Science Foundation (NSF) under Grant CNS-2030251 and EECS-2216970 and the Center for Engineering Sustainability and Resilience at Northwestern University.}}
\thanks{$^{1}$Dept. of Electrical and Computer Engineering, Northwestern University, Evanston, USA}%
\thanks{$^{2}$Dept. of Industrial Engineering and Management Sciences, Northwestern University, Evanston, USA}%
\thanks{Email addresses: \url{lihuiyi2027@u.northwestern.edu}, \url{ermin.wei@northwestern.edu}
}
}

\begin{document}

\maketitle
\thispagestyle{empty}
\pagestyle{empty}

%%%%%%%%%%%%%%%%%%%%%%%%%%%%%%%%%%%%%%%%%%%%%%%%%%%%%%%%%%%%%%%%%%%%%%%%%%%%%%%%
\begin{abstract}
    Recent advancements in vehicle autonomy have drawn interest in understanding the impact of autonomous vehicles on traffic systems. In this paper, we study a traffic assignment problem in a mixed-autonomy setting where both human-driven and autonomous vehicles coexist. We model the interaction as a simultaneous routing game where human drivers are self-interested and aim to minimize their own travel times, while autonomous agents are altruistic and aim to minimize the total social cost. {The standard nonatomic congestion game analysis establishes the existence of equilibrium to this game under convex cost functions, and does not have any implication of its uniqueness. In this work, we formulate the equilibrium as a variational inequality (VI), which enables us to establish the equilibrium existence without convexity assumption, and guarantees the uniqueness of the aggregated link flow and social cost at equilibrium under a specific class of cost functions.}
    Leveraging this VI framework, we provide sufficient conditions under which {including} autonomous agents  improves, deteriorates, or has no effect on social cost. While the possibility of deterioration has been established in prior work, our results complement existing worst-case bounds by explicitly characterizing sufficient conditions under which each outcome occurs, thereby providing a {deeper} understanding of mixed-autonomy traffic systems. Furthermore, we consider a centralized scenario where a social planner optimizes the routing of autonomous agents, and show that the same equilibrium is achieved as in the decentralized scenario {when assuming convex costs.}{Finally, we conduct numerical experiments that illustrate how social cost changes with the amount of autonomous vehicles under different system parameters.}
\end{abstract}

%%%%%%%%%%%%%%%%%%%%%%%%%%%%%%%%%%%%%%%%%%%%%%%%%%%%%%%%%%%%%%%%%%%%%%%%%%%%%%%%
\section{Introduction}

Autonomous vehicle technology promises to transform transportation systems. In a fully autonomous regime, autonomous vehicles can improve traffic efficiency, safety, and energy efficiency through precise control and coordination \cite{Levin2015}. In the near term, however, transportation networks will remain in a mixed-autonomy regime, where autonomous vehicles coexist and interact with human-driven vehicles. Therefore, understanding the equilibrium and efficiency of such mixed-autonomy systems is crucial for policy-making and infrastructure planning \cite{USDOT_AV30}.

Even at low penetration rates, autonomous vehicles have demonstrated the potential to improve system performance significantly. Early research has shown that a single autonomous agent can dampen stop-and-go waves and stabilize traffic flow locally \cite{Cui2017}. At a microscopic level, reinforcement learning (RL) has been used for autonomous vehicles to generate efficient platooning and spacing behaviors \cite{Wu2017}, while social-preference and implicit-coordination mechanisms can improve safety and efficiency in human--autonomy interactions \cite{Schwarting2019,Zhao2023,Toghi2022}. Researchers have also investigated the equilibrium of lane-changing maneuvers \cite{Mehr2021}, multi-lane merging through coalition games \cite{Fu2025}, and coupled vehicle-signal control using Stackelberg games \cite{Zhang2025}. Additionally, the integration of autonomous vehicles into mobility-on-demand services (e.g., ride-hailing) has been explored to optimize pricing and social welfare \cite{Xie2023}. While these studies provide valuable insights into specific mixed-autonomy traffic interactions and control systems, there remains a critical need to understand the macroscopic impact of mixed-autonomy on network-wide traffic assignment and equilibrium.

{At the network level, traditional transportation system with human-only traffic uses the classical Wardrop equilibrium \cite{Wardrop} to capture the traffic assignment equilibrium: each infinitesimal traveler chooses a minimum-travel-time route, so every used route has travel time no greater than any unused alternative.} However, mixed autonomy departs from this benchmark because different agents may have different information, physical characteristics, or objectives. Existing work has studied information asymmetry between human and autonomous vehicles \cite{Wang2019,Zhong2024}, heterogeneous lateral and longitudinal interactions \cite{Li2022}, and the price of anarchy when autonomous vehicles reduce congestion through shorter headways \cite{Lazar2018}. {These works enrich the classical routing model through heterogeneous information or congestion effects, but they do not assign autonomous agents a fundamentally different routing objective from the rest of the population.} {A distinct line of work instead} allows autonomous agents to behave altruistically, meaning that they may sacrifice their own travel time to improve overall system performance. For example, Biyik \textit{et al.} consider a network of parallel roads and demonstrate that even a small degree of altruism can significantly reduce total social cost \cite{Biyik2020}.

{The closest prior works to ours are \cite{Brown2020} and \cite{Hill2023}, which model non-selfish agents via a parametric cost function $t_a(f_a) + \varphi_a f_a t_a'(f_a)$, {where $t_a(f_a)$ is the travel time on link $a$ as a function of flow $f_a$, $f_a t_a'(f_a)$ is the marginal cost pricing accounting for agents' externalities}, and $\varphi_a \ge 0$ is a design parameter representing the degree of altruism or toll sensitivity. In both papers, equilibrium existence is established under convex travel times and a single origin-destination (O/D) pair, and uniqueness of the equilibrium or social cost is not studied. Brown and Marden \cite{Brown2020} consider heterogeneous toll sensitivities and prove an impossibility result for network-agnostic mechanisms: any mechanism that improves social cost on some class of networks must worsen it on others. %, so their conclusions are worst-case over classes of networks and sensitivities. 
Hill and Brown \cite{Hill2023} focus on series-parallel networks with selfish and $\varphi$-altruistic agents and characterize a robust choice of $\varphi$ that balances improvement against deterioration over that network class. These papers study worst-case or robust behavior over classes of networks and altruism levels, however, an explicit characterization of the mixed-autonomy impact on social cost for a fixed network instance and cost structure remains open.}

This paper takes a complementary perspective. For a fixed network and link-cost structure, we ask: \textit{what can be said about the equilibrium and the resulting social cost when a fraction of travelers are fully altruistic autonomous agents?} To answer this question, we formulate the mixed-autonomy routing game as a variational inequality (VI) \cite{Dafermos_VI}. %, \ly{since it is more general than the congestion game analysis.} 
{This formulation allows us to establish equilibrium existence without convexity assumption on travel times, and applies to multiple O/D pairs. Moreover, we show the uniqueness of the equilibrium aggregated link flow and social cost under the Bureau of Public Roads (BPR) travel time function, which is widely used in traffic modeling \cite{BPR_function,BPR_function2,BPR_function3,BPR_function4}.} We further derive sufficient conditions under which including autonomous agents improves, deteriorates, or has no effect on the social cost. Finally, under convex travel time functions, we show that decentralized altruistic routing yields the same equilibrium outcome as centralized routing of the autonomous fleet.

Our main contributions are summarized as follows:
\begin{itemize}
    \item We formulate the mixed-autonomy traffic assignment problem as a variational inequality and establish equilibrium existence under general continuously differentiable nondecreasing link travel times.
    \item Under BPR travel time functions, we prove uniqueness of the aggregated equilibrium link flow and the resulting equilibrium social cost, even though the VI cost mapping is not strictly monotone.
    \item We derive algebraic sufficient conditions under which including fully altruistic autonomous agents improves, deteriorates, or has no effect on the social cost, thereby complementing worst-case analyses with instance-specific characterizations.
    \item Under convex travel time functions, we prove that the equilibrium outcome induced by decentralized altruistic autonomous agents coincides with that induced by a centralized controller routing the autonomous fleet.
\end{itemize}

The remainder of this paper is organized as follows. Section \ref{sec: prelim} reviews preliminaries on variational inequalities. Section \ref{sec: model} introduces the mixed-autonomy model. Section \ref{sec: equilibrium analysis} characterizes the equilibrium using VI theory and establishes existence and uniqueness results. Section \ref{sec: impact of autonomous agents} analyzes the impact of autonomous agents on social cost. Section \ref{sec: numerical experiments} presents numerical experiments. Section \ref{sec: centralized} discusses the equivalence between decentralized and centralized control. Section \ref{sec: conclusion} concludes the paper.

Throughout the paper, we use the following notations. For a vector $x$, $x_i$ or $[x]_i$ denotes its $i$-th entry. For a matrix $A$, $A_{ij}$ denotes its $(i,j)$-th entry, $[A]_i$ denotes its $i$-th row, and $A^T$ denotes its transpose. For a set $\mathcal{S}$, $|\mathcal{S}|$ denotes its cardinality. We use $\mathbb{R}^n$ to denote the $n$-dimensional Euclidean space, and $\mathbb{R}^{n}_+$ to denote the non-negative orthant in $\mathbb{R}^n$. We use $\mathbf{0}$ and $\mathbf{1}$ to denote the all-zero and all-one vectors, respectively, with dimensions determined by the context. For two vectors $x$ and $y$ of the same dimension, we use $x\geq y$ to denote the element-wise inequality, i.e., $x_i \geq y_i$ for all $i$. The dot product of two vectors $x$ and $y$ is denoted by $x \cdot y$, which is also equivalent to $x^T y$. For two sets $\mathcal{S}_1$ and $\mathcal{S}_2$, we use $\mathcal{S}_1 \setminus \mathcal{S}_2$ to denote the set difference, i.e., the set of elements in $\mathcal{S}_1$ but not in $\mathcal{S}_2$.

Due to space limitations, some proofs are relegated to the appendix.

\section{Preliminaries: Variational Inequality}\label{sec: prelim}

In this section, following \cite{nagurney_2013_network}, we briefly review the variational inequality (VI) framework and several standard results that will be used in our equilibrium analysis.
\begin{definition}[Variational Inequality]
    The finite-dimensional variational inequality problem, $VI(F,\mathcal{K})$, is to determine a vector $x^* \in \mathcal{K} \subset \mathbb{R}^n$, such that 
    \begin{equation*}
        F(x^*) \cdot (x - x^*) \geq 0, \quad \forall x \in \mathcal{K},
    \end{equation*}
    where $F$ is a given continuous function from $\mathcal{K}$ to $\mathbb{R}^n$ and $\mathcal{K}$ is a given closed convex set.
\end{definition}

We next recall two standard results on existence and uniqueness.

\begin{proposition}\label{prop: VI existence}
The VI problem $VI(F,\mathcal{K})$ has at least one solution if $\mathcal{K}$ is compact and convex, and $F$ is continuous on $\mathcal{K}$.
\end{proposition}

\begin{proposition}\label{prop: VI uniqueness}
The solution to the VI problem $VI(F,\mathcal{K})$ is unique, provided that $F$ is strictly monotone on $\mathcal{K}$, i.e.,
\begin{equation*}
    \big[ F(x_1) - F(x_2) \big] \cdot (x_1 - x_2) > 0, \quad \forall x_1, x_2 \in \mathcal{K}, x_1 \neq x_2.
\end{equation*}
\end{proposition}

Many mathematical problems can be formulated as VI problems, including systems of equations, optimization problems, complementarity problems, fixed point problems, etc. For example, consider an optimization problem $\min_{x \in \mathcal{K}} \, f(x)$, where $f$ is continuously differentiable and convex, and $\mathcal{K}$ is a closed convex set. Then, $x^*$ is a solution to this problem if and only if it is a solution to the VI problem $VI(\nabla f, \mathcal{K})$.

\section{Model}\label{sec: model}
In this section, we introduce a nonatomic traffic assignment model with mixed autonomy. We consider a transportation network represented by a directed graph $\mathcal{G}=(\mathcal{N},\mathcal{L})$, where $\mathcal{N}$ is the set of nodes and $\mathcal{L}$ is the set of directed links. Links are denoted by $a, b, c$, and $|\mathcal{L}|=L$. We assume there is one origin-destination (O/D) pair, representing the origin and destination of the traffic flow. The set of paths connecting the O/D pair is denoted by $\mathcal{P}$. We assume $\mathcal{P}$ is non-empty, otherwise the O/D pair is not feasible. Paths are denoted by $p, q$, and $|\mathcal{P}|=P$. 

Suppose the O/D pair has a fixed unit travel demand, which is the total flow that needs to be routed from the origin to the destination. We consider two types of agents (vehicles) in the network: human agents and autonomous agents. Let $\alpha \in [0,1]$ denote the fraction of autonomous agents. Then, the demands of human and autonomous agents are $1-\alpha$ and $\alpha$, respectively.

We denote the flow of human and autonomous agents on path $p \in \mathcal{P}$ by $x_p^H$ and $x_p^A$, respectively. The aggregated flow on path $p$ is then $x_p = x_p^H + x_p^A$. The flow needs to meet the demand of the O/D pair, which leads to the following constraints:
\begin{equation}\label{eq: flow meets demand}
\sum_{p \in \mathcal{P}} x_p^H = 1-\alpha, \quad \sum_{p \in \mathcal{P}} x_p^A = \alpha.
\end{equation}
We group the path flows into column vectors $x^H = (x_p^H)_{p \in \mathcal{P}}$, $x^A = (x_p^A)_{p \in \mathcal{P}}$, and $x = (x_p)_{p \in \mathcal{P}}$. 

For each link $a \in \mathcal{L}$, let $f_a^H$ and $f_a^A$ denote the flow of human and autonomous agents on link $a$, respectively. The aggregated flow on link $a$ is then $f_a = f_a^H + f_a^A$. We group the link flows into column vectors $f^H = (f_a^H)_{a \in \mathcal{L}}$, $f^A = (f_a^A)_{a \in \mathcal{L}}$, and $f = (f_a)_{a \in \mathcal{L}}$. The link flows are related to the path flows via the link-path incidence matrix \cite{incidence_matrix}. Specifically, let $\Delta \in \{0,1\}^{L \times P}$ denote the link-path incidence matrix, where $\Delta_{ap} = 1$ if link $a$ is on path $p$, and $\Delta_{ap} = 0$ otherwise. Then, we have
\begin{equation}\label{eq: link-path relation}
f^H = \Delta x^H, \quad f^A = \Delta x^A.
\end{equation}
Note that once the path flows $x^H$ and $x^A$ are determined, the link flows $f^H$ and $f^A$ are also determined. However, the opposite direction is not necessarily true. One can have multiple path flow patterns that lead to the same link flow pattern. The set of feasible link flow patterns is given by
\begin{equation*}
\mathcal{K} = \left\{ (f^H, f^A) \; | \; \exists x^H, x^A \ge \mathbf{0} \text{ satisfying } \eqref{eq: flow meets demand} \text{ and } \eqref{eq: link-path relation} \right\}.
\end{equation*}
Since the feasible set of $(x^H, x^A)$ is compact and convex, and the mapping from $(x^H, x^A)$ to $(f^H, f^A)$ is linear, the set $\mathcal{K}$ is also compact and convex.

Next, we define the \textit{travel time} functions for the links. The travel time on link $a$ depends only on the aggregated link flow $f_a$, regardless of agent types. We denote the travel time function on link $a$ by $t_a(f_a)$, which is assumed to be continuously differentiable, non-decreasing, and non-negative. 

Human agents are self-interested and minimize their own travel times. Thus, the cost (disutility) for human agents on path $p$ is given by
\begin{equation*}
C_p^H(x) = \sum_{a \in \mathcal{L}} \Delta_{ap} t_a(f_a).
\end{equation*}
Autonomous agents are altruistic and evaluate routes according to the total \textit{social cost}, defined as the total travel time incurred by all agents in the network \cite{Wardrop},
\begin{equation*}
S(x) = \sum_{p \in \mathcal{P}} \Big( x_p \sum_{a \in \mathcal{L}} \Delta_{ap} t_a(f_a) \Big) = \sum_{a \in \mathcal{L}} f_a t_a(f_a),
\end{equation*}
where the second equality follows from \eqref{eq: link-path relation}. Note that both $C_p^H(x)$ and $S(x)$ depend on the path flows $x$ only through the link flows $f$. Thus, we can also write them as $C_p^H(f^H,f^A)$ and $S(f^H,f^A)$, or simply $C_p^H(f)$ and $S(f)$.

We focus first on the decentralized setting, in which autonomous agents are not controlled by a central planner. In Section \ref{sec: centralized}, we will show that, under convex travel time functions, the equilibrium outcome in this decentralized setting coincides with that of a corresponding centralized routing problem for the autonomous flow.

We study the simultaneous routing game in which human and autonomous agents choose their paths to minimize their respective objectives. A feasible path flow pattern $(x^{H*}, x^{A*})$ is an \textit{equilibrium} if no (infinitesimal) agent can improve their objective by unilaterally switching to a different path. Specifically:
\begin{enumerate}
    \item No human agent can reduce their travel time by unilaterally switching paths;
    \item No autonomous agent can decrease the social cost by unilaterally switching paths.
\end{enumerate}
The induced link flow pattern $(f^{H*}, f^{A*}) = (\Delta x^{H*}, \Delta x^{A*})$ is called an equilibrium link flow pattern. Unless otherwise specified, we refer to the equilibrium link flow pattern when we mention equilibrium in the rest of the paper. The formal $\epsilon$-deviation definition is given in Appendix \ref{appendix: proof of thm 1}.

The first condition is the classical Wardrop equilibrium \cite{Wardrop}: all paths used by human agents have equal travel time, and no unused path has a lower travel time. That is, for every path $p \in \mathcal{P}$,
\begin{align*}
    C_p^H(f^{H*}, f^{A*}) &\begin{cases}= \lambda^H, & \text{if } x_p^{H*} > 0, \\ \geq \lambda^H, & \text{if } x_p^{H*} = 0,\end{cases}
\end{align*}
where $\lambda^H$ is the equilibrium travel time for human agents.

The second condition involves the social cost, which is a global objective rather than a path-based cost. Unlike the human agents' condition, it does not directly reduce to a Wardrop-like principle. We will show in Section \ref{sec: equilibrium analysis} (Theorem \ref{thm: social cost to path cost}) that this condition is equivalent to a Wardrop-like condition with respect to a path-based marginal social cost.

\section{Equilibrium Existence, Uniqueness, and Computation}\label{sec: equilibrium analysis}
In this section, we characterize the equilibrium of the mixed-autonomy game, formulate it as a variational inequality (VI) problem, and establish existence and uniqueness results.

We begin by addressing the autonomous agents' equilibrium condition. Recall that each autonomous agent seeks to minimize the social cost, a global objective. The following result shows that this condition induces a Wardrop-like condition with respect to path-based marginal social costs. Define
\begin{align*}
    C_p^A(f^H, f^A) = \sum_{a\in L}\Delta_{ap}\,\frac{d\bigl(f_a t_a(f_a)\bigr)}{d f_a}.
\end{align*}

\begin{theorem}\label{thm: social cost to path cost}
The equilibrium link flow pattern $(f^{H*}, f^{A*}) \in \mathcal{K}$, with corresponding path flow pattern $(x^{H*}, x^{A*})$, is characterized by the following condition: there exists scalars $\lambda^H$ and $\lambda^A$ such that, for every path $p \in \mathcal{P}$,
\begin{align*}
    C_p^H(f^{H*}, f^{A*}) &\begin{cases}= \lambda^H, & \text{if } x_p^{H*} > 0, \\ \geq \lambda^H, & \text{if } x_p^{H*} = 0,\end{cases} \\
    C_p^A(f^{H*}, f^{A*}) &\begin{cases}= \lambda^A, & \text{if } x_p^{A*} > 0, \\ \geq \lambda^A, & \text{if } x_p^{A*} = 0,\end{cases}
\end{align*}
where $\lambda^H$ and $\lambda^A$ denote the equilibrium costs for human agents and autonomous agents, respectively.
\end{theorem}

\begin{proof}
The first condition is the standard Wardrop condition for human travelers. For the second condition, consider an infinitesimal deviation of autonomous flow from one path to another and compute the resulting change in the social cost. This yields the marginal social cost along each path. The full proof is given in Appendix \ref{appendix: proof of thm 1}.
\end{proof}

The quantity $C_p^A(f^{H}, f^{A})$ serves as the cost/disutility for autonomous agents. It can be interpreted as the marginal social cost on path $p$, which measures the increase in total social cost generated by assigning an additional infinitesimal amount of flow to path $p$. At the link level, define
$$c_a^H(f)=t_a(f_a), \; c_a^A(f) = \frac{d(f_a t_a(f_a))}{df_a} = t_a(f_a) + f_a t_a'(f_a).$$
Thus, human agents perceive the physical travel time, while autonomous agents perceive the travel time plus the congestion externality they impose on all users.

This cost function of autonomous agents is closely related to {marginal cost pricing} (Pigouvian toll) \cite{Beckmann1956}: in a standard single-class congestion game, charging every user a toll $f_a t_a'(f_a)$ causes each user to perceive the marginal social cost, and the resulting Wardrop equilibrium coincides with the social optimum. However, the equivalence established in Theorem \ref{thm: social cost to path cost} is not a straightforward application of this classical result, as two types of agents coexist in our setting. Furthermore, the classical Pigouvian guarantee that marginal cost pricing achieves the social optimum does not hold in our setting, because only a fraction $\alpha$ of agents perceive the marginal social cost while the remaining fraction $1-\alpha$ remain self-interested. The net effect on social cost therefore depends on the network structure, as we analyze in Section \ref{sec: impact of autonomous agents}.

By Theorem \ref{thm: social cost to path cost}, mixed-autonomy routing game can be viewed as a two-class traffic assignment problem \cite{Dafermos_multiclass} with link costs $c^H$ and $c^A$. This allows us to formulate the equilibrium as a VI \cite{nagurney_2013_network}. We group the link costs into column vectors $c^H(f) = (c_a^H(f))_{a \in \mathcal{L}}$ and $c^A(f) = (c_a^A(f))_{a \in \mathcal{L}}$. The path costs are then $C_p^H(f) = \sum_{a \in \mathcal{L}} \Delta_{ap} c_a^H(f)$ and $C_p^A(f) = \sum_{a \in \mathcal{L}} \Delta_{ap} c_a^A(f)$.

\begin{lemma}\label{lemma: VI formulation}
A link flow pattern $(f^{H*}, f^{A*}) \in \mathcal{K}$ is an equilibrium if and only if it satisfies the following VI problem:
\begin{equation}\label{eq: VI}
c(\Bar{f}^*) \cdot (\Bar{f} - \Bar{f}^*) \geq 0, \quad \forall (f^H, f^A) \in \mathcal{K},
\end{equation}
where $\Bar{f} = \begin{bmatrix} f^H \\ f^A \end{bmatrix}$, $\Bar{f}^* = \begin{bmatrix} f^{H*} \\ f^{A*} \end{bmatrix}$, and $c(\Bar{f}^*) = \begin{bmatrix} c^H(\Bar{f}^*) \\ c^A(\Bar{f}^*) \end{bmatrix}$.
\end{lemma}

This Lemma follows directly from Theorem \ref{thm: social cost to path cost} and the literature \cite{Dafermos_VI}. Lemma \ref{lemma: VI formulation} immediately yields existence.

\begin{corollary}[Equilibrium Existence]\label{cor: equilibrium existence}
    There exists at least one equilibrium link flow pattern $(f^{H*}, f^{A*}) \in \mathcal{K}$.
\end{corollary}
The proof follows from Proposition \ref{prop: VI existence}, since $\mathcal{K}$ is compact and convex, and $c$ is continuous on $\mathcal{K}$.

\begin{remark}\label{rmk: generalization}
The VI-based existence argument requires only that $c$ be continuous on a compact convex set, which holds whenever the travel time functions $t_a(\cdot)$ are continuously differentiable and nondecreasing. In comparison, the existence results in related works \cite{Brown2020, Hill2023} rely on convexity of $t_a(\cdot)$, which is a stronger assumption. Our VI formulation thus establishes equilibrium existence under weaker conditions. Furthermore, Theorem \ref{thm: social cost to path cost}, Lemma \ref{lemma: VI formulation}, and Corollary \ref{cor: equilibrium existence} all extend to networks with multiple O/D pairs---a generalization not addressed in \cite{Brown2020} or \cite{Hill2023}, both of which restrict to a single O/D pair. The proof for the general case is provided in Appendix \ref{appendix: proof of thm 1}.
\end{remark}

Although existence is guaranteed, uniqueness generally fails. The VI problem \eqref{eq: VI} may admit multiple solutions when the cost mapping $c$ is not strictly monotone (Proposition \ref{prop: VI uniqueness}). Importantly, $c$ can fail to be strictly monotone even when both $c_a^H$ and $c_a^A$ are individually strictly increasing.

A simple example illustrates the non-uniqueness: consider a network with two parallel links having identical travel time functions. Any flow pattern that splits the total demand equally between the two links is an equilibrium, regardless of how the flow is allocated between human and autonomous agents on each link.

We next show that, although the individual flows $f^{H*}$ and $f^{A*}$ may not be unique, the aggregated link flow $f^* = f^{H*} + f^{A*}$ is unique under a standard class of travel time functions. Specifically, we impose the following assumption.
\begin{assumption}\label{assump: BPR}
For every link $a\in \mathcal{L}$, the travel time function is of the form
\begin{equation*}
t_a(f_a)=k_a f_a^n+b_a,
\end{equation*}
where $k_a>0$, $b_a\ge 0$, and $n\ge 1$.
\end{assumption}
Assumption \ref{assump: BPR} includes the standard Bureau of Public Roads (BPR) travel time function \cite{BPR_function,BPR_function2,BPR_function3,BPR_function4}, which is widely used in traffic modeling. The BPR function is given by
\begin{equation*}
t_a(f_a)=t_a^0\left(1+\theta\left(\frac{f_a}{m_a}\right)^\beta\right),
\end{equation*}
where $t_a^0$ is the free-flow travel time (travel time with zero traffic) on link $a$, $m_a$ is the capacity of link $a$, and $\theta, \beta$ are parameters, commonly set as $\theta = 0.15$ and $\beta = 4$ \cite{BPR_function3}. By setting $k_a = t_a^0 \theta / m_a^{\beta}$, $n = \beta$, and $b_a = t_a^0$, we can see that the BPR function fits into the aforementioned class of functions.

Under Assumption \ref{assump: BPR}, we have the following uniqueness result.

\begin{theorem}[Uniqueness of Aggregated Link Flow]\label{thm: uniqueness of aggregated link flow}
Under Assumption \ref{assump: BPR}, the aggregated link flow at equilibrium $f^* = f^{H*} + f^{A*}$ is unique.
\end{theorem}

Consequently, the equilibrium social cost is also unique.

\begin{corollary}[Uniqueness of Social Cost]
Under Assumption \ref{assump: BPR}, the social cost at equilibrium $S(f^{H*}, f^{A*})$ is unique.
\end{corollary}
This follows directly from Theorem \ref{thm: uniqueness of aggregated link flow}, since the social cost depends on the flow pattern only through the aggregated link flow $f^*$.

Finally, we comment on computation. Since the equilibrium is characterized by the VI problem \eqref{eq: VI}, one may use standard iterative methods for VI problems \cite{nagurney_2013_network}. A natural approach is Gauss--Seidel-type relaxation scheme: starting from an initial feasible flow pattern, alternately solve for the human equilibrium with the autonomous flow fixed, and then solve for the autonomous equilibrium with the human flow fixed. Each subproblem is a standard single-class traffic assignment problem, and can be handled efficiently using methods such as the Frank--Wolfe algorithm \cite{nagurney_2013_network}. Convergence is guaranteed when the cost mapping $c$ is strongly monotone \cite{nagurney_2013_network}. However, in our setting, $c$ is generally not strongly monotone, hence we do not claim general convergence. Nevertheless, in our numerical experiments in Section \ref{sec: numerical experiments}, the method converges reliably. 

\section{Impact of Autonomous Agents on the System}\label{sec: impact of autonomous agents}

In this section, we study how including autonomous agents changes the equilibrium social cost. We derive sufficient conditions under which the social cost improves, deteriorates, or remains unchanged. Throughout this section, Assumption \ref{assump: BPR} is in force.

We compare two systems. The \textit{baseline system} has no autonomous agents ($\alpha = 0$). We denote its equilibrium link flow pattern by $(\tilde{f}^{*}, \mathbf{0})$, equilibrium path flow pattern by $(\tilde{x}^{*}, \mathbf{0})$, equilibrium human cost by $\tilde{\lambda}^H$, and social cost by $\tilde{S}^*$. The \textit{mixed-autonomy system} has a positive fraction of autonomous agents ($\alpha > 0$). We denote its equilibrium link flow pattern by $(f^{H*}, f^{A*})$, equilibrium path flow pattern by $(x^{H*}, x^{A*})$, equilibrium human cost and autonomous cost by $\lambda^H$ and $\lambda^A$, and social cost by $S^*$.

Note that both systems may admit multiple equilibrium link flow patterns. However, Theorem \ref{thm: uniqueness of aggregated link flow} implies that the equilibrium aggregated link flow is unique in both systems: $\tilde{f}^{*}$ in the baseline system and $f^* = f^{H*} + f^{A*}$ in the mixed-autonomy system. Since the social cost depends only on the aggregated link flow, $\tilde{S}^*$ and $S^*$ are also uniquely determined.

\subsection{Improvements}

Intuitively, autonomous agents may reduce the social cost, since they aim to minimize it. We begin with a sufficient condition for improvement that applies to an arbitrary network.

\begin{lemma}\label{lemma: improvement condition}
If there exists a baseline equilibrium path flow $\tilde{x}^*$ and a mixed-autonomy equilibrium path flow $(x^{H*}, x^{A*})$ such that $x^{H*} \le \tilde{x}^{*}$, then $S^* \le \tilde{S}^*$. The inequality is strict if, in addition, $f^{*} \neq \tilde{f}^{*}$.
\end{lemma}

The condition $x^{H*} \le \tilde{x}^*$ means that, path by path, the human flow in the mixed-autonomy equilibrium does not exceed its baseline value. Thus, autonomous flow can be interpreted as replacing part of the baseline human flow rather than displacing human agents onto less favorable routes. Whenever this occurs, the social cost cannot increase. Since equilibrium need not be unique but the social cost is unique, the condition is stated in existential form: it is enough that one baseline equilibrium path flow and one mixed-autonomy equilibrium path flow satisfy the inequality.

This criterion is useful in two ways. First, it is topology-independent: whenever the condition holds for a given network instance, improvement is guaranteed. Second, it can be used as a proof technique for classes of network topologies by showing that the condition is always satisfied on that class.

{In particular, we find that social cost is always improved on networks consisting of parallel links connected in series. We formalize this network class as follows.}

\begin{definition}[Path Multigraph]\label{def: path multigraph}
A path multigraph is a directed graph whose nodes can be ordered as $v_1, v_2, \ldots, v_N$ such that every path from the origin $v_1$ to the destination $v_N$ traverses the same sequence of nodes, and paths may differ only in the links connecting consecutive nodes.
\end{definition}

{To show this guaranteed improvement, one can constructively show that the condition $x^{H*} \le \tilde{x}^*$ in Lemma \ref{lemma: improvement condition} is always satisfied on path multigraphs by building a baseline equilibrium from the mixed-autonomy equilibrium. We provide this constructive argument in Appendix \ref{appendix: proof of thm 3}. Note that our finding coincides with literature \cite{Brown2020}, but using different proving techniques.}

\subsection{Deteriorations}

The preceding result shows that autonomous agents improve social cost on path multigraphs. However, this is not always the case. It is known from prior work that altruistic or toll-responsive agents can increase social cost on certain networks: Hill and Brown \cite{Hill2023} show that on some series-parallel networks, including any positive fraction of altruistic agents deteriorates social cost, and Brown and Marden \cite{Brown2020} prove that any network-agnostic mechanism that improves social cost on some networks must deteriorate it on others. In this subsection, we complement these worst-case results by deriving a sufficient condition for deterioration on a given network.

First, we introduce some notation. In the baseline system, we denote $\mathcal{V}$ as the set of used paths at equilibrium, i.e., 
\begin{equation*}
    \mathcal{V} = \{ p \in \mathcal{P} \; | \; \tilde{x}_p^* > 0 \}.
\end{equation*}
The number of used paths is then $|\mathcal{V}| = V$. Let $\Delta_{\mathcal{V}} \in \{0,1\}^{L \times V}$ be the link-path incidence matrix restricted to the used paths, i.e., $\Delta_{\mathcal{V}}$ only contains the columns of $\Delta$ corresponding to the paths in $\mathcal{V}$. Let $M_{\mathcal{V}} = \Delta_{\mathcal{V}}^T K \Delta_{\mathcal{V}}$, where $K = \text{diag}(k)$ is a diagonal matrix with parameters $k_a$ on the diagonal. It has the following property:

\begin{lemma}\label{lemma: M invertible}
    $M_{\mathcal{V}}$ is invertible if and only if $\Delta_{\mathcal{V}}$ has linearly independent columns.
\end{lemma}
\begin{proof}
    Since $M_{\mathcal{V}}$ is a square matrix, it is invertible if and only if its null space contains only the zero vector. Suppose $\Delta_{\mathcal{V}}$ has linearly independent columns, we need to show that $M_{\mathcal{V}} y = \mathbf{0}$ implies $y = \mathbf{0}$. Note that $M_{\mathcal{V}} y = \Delta_{\mathcal{V}}^T K \Delta_{\mathcal{V}} y = \mathbf{0}$,
    which implies that 
    \begin{equation*}
        y^T \Delta_{\mathcal{V}}^T K \Delta_{\mathcal{V}} y = (\Delta_{\mathcal{V}} y)^T K (\Delta_{\mathcal{V}} y) = 0.
    \end{equation*}
    Let $z = \Delta_{\mathcal{V}} y$, then we obtain
    \begin{equation*}
        z^T K z = \sum_{a \in \mathcal{L}} k_a z_a^2 = 0.
    \end{equation*}
    Since $k_a > 0$ for all $a \in \mathcal{L}$, we have $z_a = 0$ for all $a \in \mathcal{L}$, i.e., $z = \mathbf{0}$. Since $\Delta_{\mathcal{V}}$ has linearly independent columns, we have $y = \mathbf{0}$. 

    On the other hand, suppose $\Delta_{\mathcal{V}}$ has linearly dependent columns, we need to show that there exists a non-zero vector $y$ such that $M_{\mathcal{V}} y = \mathbf{0}$. Let $y$ be a non-trivial solution to $\Delta_{\mathcal{V}} y = \mathbf{0}$. Then, we have
    \begin{equation*}
        M_{\mathcal{V}} y = \Delta_{\mathcal{V}}^T K \Delta_{\mathcal{V}} y = \Delta_{\mathcal{V}}^T K \mathbf{0} = \mathbf{0}.
    \end{equation*}
    Thus, the null space of $M_{\mathcal{V}}$ contains a non-zero vector $y$, and therefore, $M_{\mathcal{V}}$ is not invertible.
\end{proof}

The condition that $\Delta_{\mathcal V}$ has linearly independent columns is purely linear-algebraic. A stronger but more transparent sufficient condition is that each path in $\mathcal V$ contains at least one link that does not belong to any other path in $\mathcal V$.

\begin{theorem}\label{thm: deterioration condition}
    Assume each link travel time $t_a(\cdot)$ is linear, $\Delta_{\mathcal{V}}$ has linearly independent columns, $q = \arg\min_{p \in \mathcal{P}} C_p^A(\tilde{x}^*)$ is unique and satisfies $q \notin \mathcal{V}$, and $C_p^H(\tilde{x}^*) > \tilde{\lambda}^H$ for any $p \notin \mathcal{V}$. 
    Let $e_q$ denote a unit vector with 1 at the $q$-th entry and 0 elsewhere, and define 
    $$\gamma = \frac{\mathbf{1}^T M_{\mathcal V}^{-1} \Delta_{\mathcal{V}}^T K \Delta e_q - 1}{\mathbf{1}^T M_{\mathcal V}^{-1} \mathbf{1}}.$$
    If 
    \begin{align}
        C_q^H(\tilde{x}^*) - \tilde{\lambda}^H + \gamma > 0,
    \end{align}
    then there exists $\alpha_0 > 0$ such that, for all $\alpha \in (0,\alpha_0)$, the social cost deteriorates, i.e., $S^* > \tilde{S}^*$.
\end{theorem}

\begin{proof}
Under the stated assumptions, the equilibrium can be explicitly characterized for sufficiently small $\alpha$. This yields a closed-form expression for the perturbed equilibrium social cost, which can then be compared with the baseline value $\tilde S^*$. The full derivation is given in Appendix \ref{appendix: proof of thm 4}.
\end{proof}

The condition in Theorem \ref{thm: deterioration condition} has a natural interpretation. The term $C_q^H(\tilde{x}^*) - \tilde{\lambda}^H$ measures how much more costly path $q$ is than the paths used at the baseline equilibrium from the human agents' perspective. In other words, it quantifies the travel-time sacrifice required when autonomous agents are drawn to path $q$ by the social objective. The term $\gamma$ captures the corresponding change in the equilibrium human travel cost induced by this perturbation. If the sum of these two effects is positive, then the resulting change in total social cost is positive for sufficiently small $\alpha$, and deterioration follows.

\subsection{No Effect}

Finally, we identify a condition under which autonomous agents do not affect the equilibrium social cost.

\begin{theorem}\label{thm: no effect condition}
    If all paths have the same free-flow travel time (travel time with zero traffic), i.e., $\sum_{a \in \mathcal{L}} \Delta_{ap} b_a$ is the same for all $p \in \mathcal{P}$, equivalently $\Delta^T b = b_0 \mathbf{1}$ for some constant $b_0$, then $S^* = \tilde{S}^*$.
\end{theorem}

\begin{proof}
Under this condition, the autonomous agent cost $C_p^A(f) = (n+1)\sum_{a \in \mathcal{L}} \Delta_{ap} k_a f_a^n + b_0$ is a positive affine transformation of the human agent cost $C_p^H(f) = \sum_{a \in \mathcal{L}} \Delta_{ap} k_a f_a^n + b_0$. Both agent types therefore rank paths identically, and the mixed-autonomy equilibrium has the same aggregated link flow as the baseline system. Hence, the resulting social cost is unchanged. For details, see Appendix \ref{appendix: proof of thm 5}.
\end{proof}

{Note that the condition in Theorem \ref{thm: no effect condition} is driven by the homogeneity of path free-flow travel times, rather than by any particular network topology. This indicates that, whenever all paths have the same free-flow travel time, regardless of how complicated the network is, including autonomous agents is neutral in terms of system efficiency: it neither improves nor worsens the equilibrium social cost.}

\section{Numerical Experiments}\label{sec: numerical experiments}

In this section, we present a numerical example to illustrate how the equilibrium changes with the penetration rate of autonomous agents. We consider a modified Braess's network \cite{Braess}, shown in Figure \ref{fig: braess network}. Node S is the origin and node T is the destination. The network contains two intermediate nodes, A and B. In addition to the five standard links in the Braess's network, we add a direct link from S to T.

The links are labeled as follows: link 1 connects S and A, link 2 connects A and T, link 3 connects S and B, link 4 connects B and T, link 5 connects A and B, and link 6 connects S and T directly. Path 1 contains links 1 and 2, path 2 contains links 3 and 4, path 3 contains links 1, 5, and 4, path 4 contains links 3, 5, and 2, and path 5 contains link 6 only.

We assume the travel time functions on link $a$ are given by $t_a(f_a) = k_a f_a + b_a$, where the parameters are listed in Table \ref{table: link parameters}.

For each value of $\alpha \in [0,1]$, we compute an equilibrium using the relaxation method described in Section \ref{sec: equilibrium analysis}: we fix the autonomous flow and solve the human subproblem, then fix the human flow and solve the autonomous subproblem, and iterate until convergence. In our experiments, the algorithm converges reliably for all tested values of $\alpha$.

% \vspace{-6mm}
\begin{figure}[ht]
    \centering
    \begin{minipage}{0.56\linewidth}
        \centering
        \begin{tikzpicture}
            \node[shape=circle,draw=black] (S) at (0,0) {S};
            \node[shape=circle,draw=black] (A) at (2,0.8) {A};
            \node[shape=circle,draw=black] (B) at (2,-0.8) {B};
            \node[shape=circle,draw=black] (T) at (4,0) {T};
        
            \path [->] (S) edge node[above] {1} (A);
            \path [->] (S) edge node[below] {3} (B);
            \path [->] (A) edge node[above] {2} (T);
            \path [->] (B) edge node[below] {4} (T);
            \path [->] (A) edge node[right] {5} (B);
            \path [->] (S) edge [bend right=80] node[below] {6} (T);
        \end{tikzpicture}
        \vspace{-2mm}
        \caption{A modified Braess's network.}
        \label{fig: braess network}
    \end{minipage}%
    \begin{minipage}{0.4\linewidth}
        \centering
        \par\vspace{10pt}
        % Change caption type to table for this minipage
        \makeatletter\def\@captype{table}\makeatother
        \begin{tabular}{c|c|c}
            \hline
            Link & $k_a$ & $b_a$ \\
            \hline
            1 & 10 & 1 \\
            2 & 5 & 8 \\
            3 & 2 & 7 \\
            4 & 5 & 1 \\
            5 & 3 & 1 \\
            6 & 7 & 11 \\
            \hline
        \end{tabular}
        \caption{Travel time parameters.}
        \label{table: link parameters}
    \end{minipage}
\end{figure}
% \vspace{-6mm}

\vspace{-6mm}
\begin{figure}[ht]
    \centering
    \subfigure[Social cost]{
        \begin{minipage}[t]{0.5\linewidth}
        \centering
        \includegraphics[width=\linewidth]{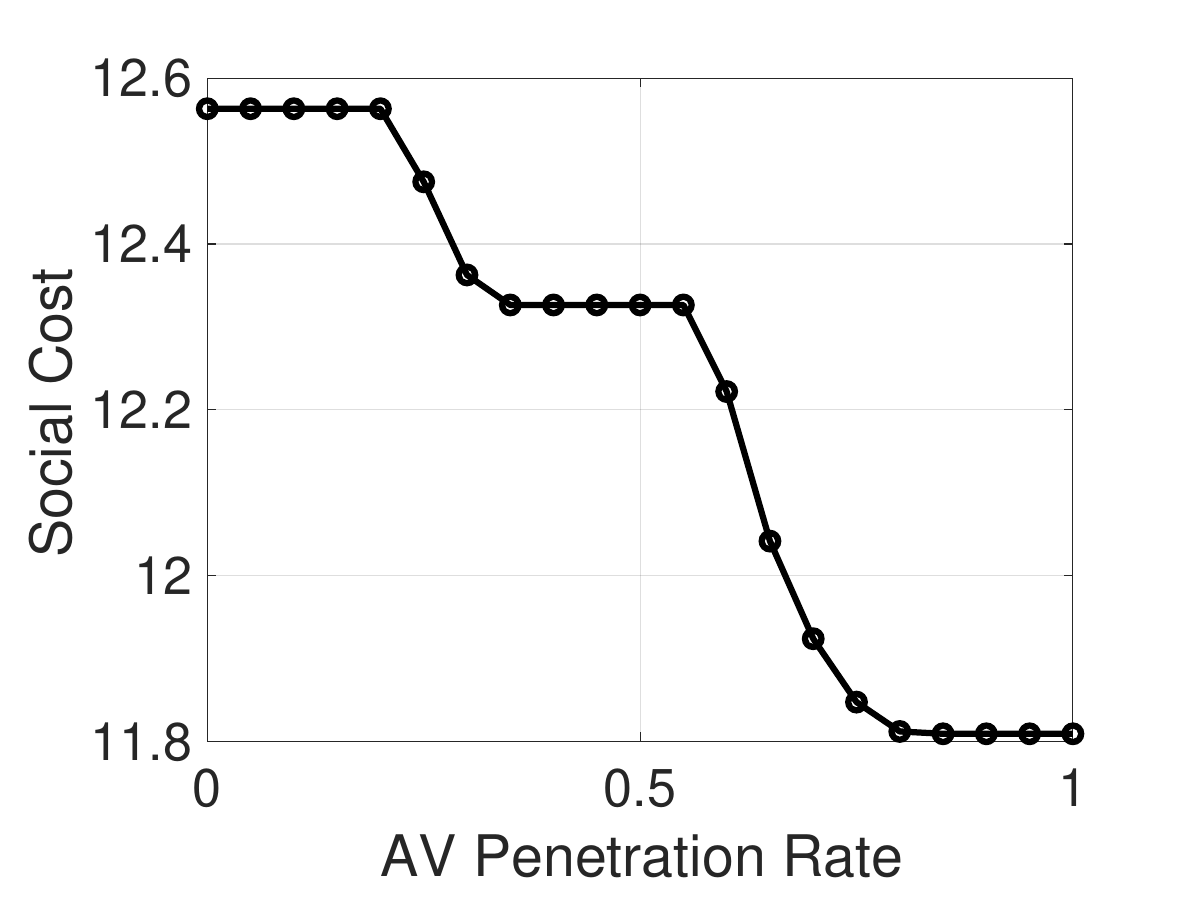}
    \end{minipage}
    }%
    \subfigure[Aggregated path flow]{
    \begin{minipage}[t]{0.5\linewidth}
        \centering
        \includegraphics[width=\linewidth]{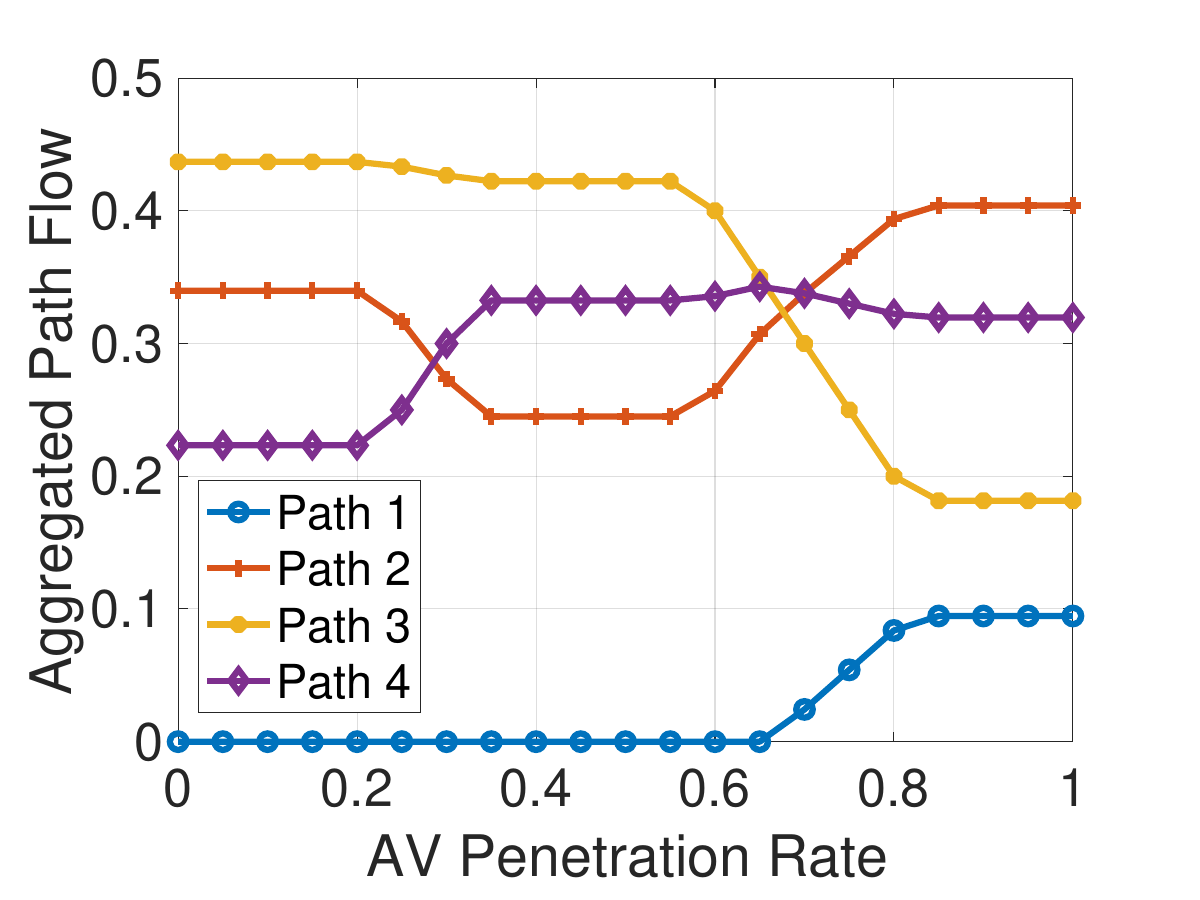}
    \end{minipage}
    }
    \vspace{-4mm}
    \caption{Social cost and aggregated path flow vs. fraction of autonomous agents.}
    \label{fig: better off}
\end{figure}
\vspace{-3mm}

Figure \ref{fig: better off}(a) shows the equilibrium social cost as a function of the autonomous penetration rate $\alpha$, while Figure \ref{fig: better off}(b) shows the equilibrium aggregated path flow. In this example, the equilibrium social cost decreases overall as $\alpha$ increases, indicating that autonomous agents improve system performance on this instance.

Several additional observations are worth noting. First, the social cost is piecewise constant over certain intervals of $\alpha$. On these intervals, the aggregated equilibrium flow remains unchanged even though the decomposition into human and autonomous flows may vary. Second, as $\alpha$ increases, Path 1, which is unused in the all-human baseline, becomes active. This shows that altruistic autonomous agents may induce the system to use routes that self-interested users alone would not select.

This result aligns with intuition: with more autonomous agents aiming to minimize the social cost, the overall system performance improves. However, as established by \cite{Brown2020} and \cite{Hill2023}, this is not always the case. While those works demonstrate the existence of deterioration, we next analyze how the social cost evolves with the fraction of autonomous agents on our network. We change the travel time function of link 6 to $t_6(f_6) = f_6 + 18.3$. Figure \ref{fig: worse off} shows the equilibrium social cost as a function of $\alpha$ for small $\alpha$. The social cost increases as $\alpha$ grows from $0$ to around $0.05$, exhibiting non-monotone behavior. This example confirms the deterioration predicted by Theorem \ref{thm: deterioration condition}. Although the VI formulation and the frameworks in \cite{Brown2020} are generally hard to solve in closed form for a given instance, the algebraic sufficient condition in Theorem \ref{thm: deterioration condition} is easy to verify and serves as a practical check for predicting deterioration at small penetration rates.

\vspace{-5mm}
\begin{figure}[ht]
    \centering
    \includegraphics[width=0.5\linewidth]{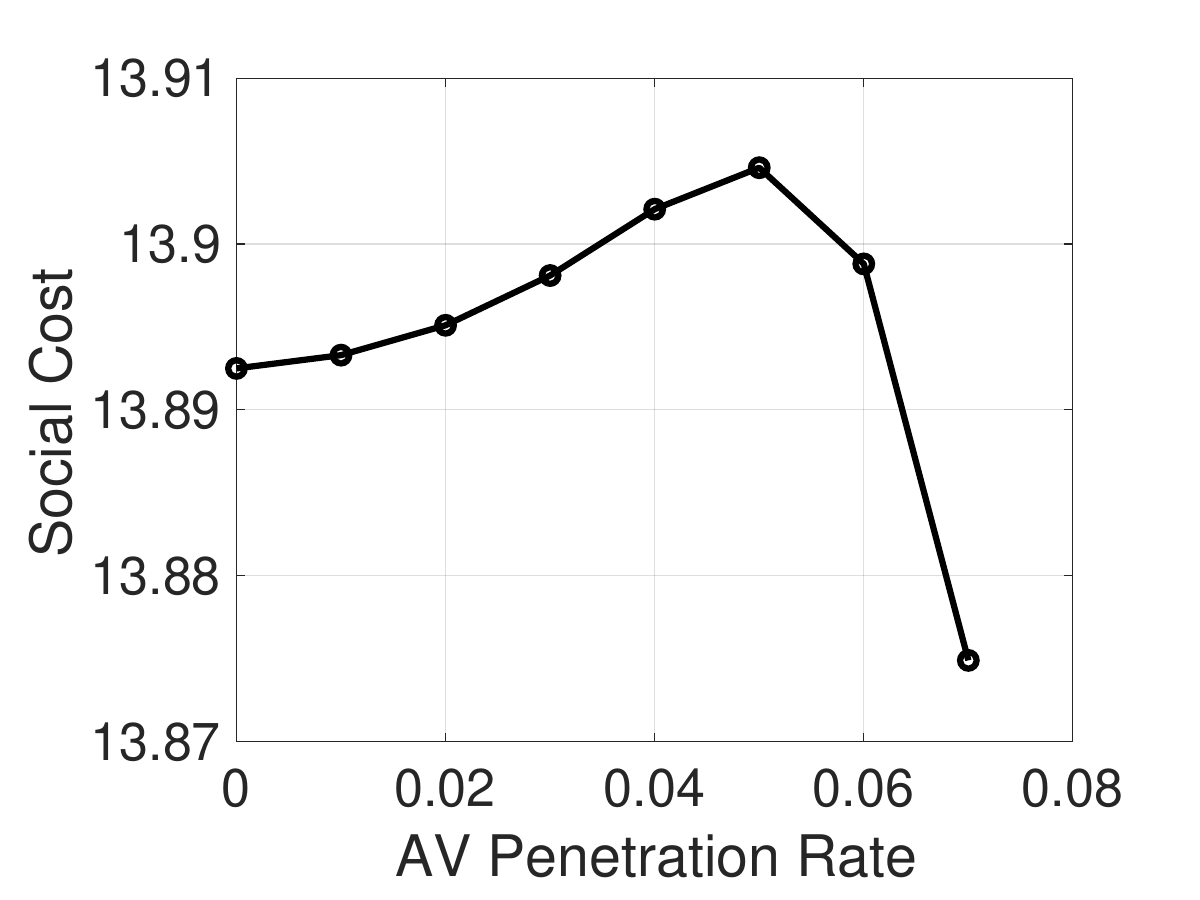}
    \vspace{-4mm}
    \caption{Social cost vs. fraction of autonomous agents (deterioration case, $b_6 = 18.3$).}
    \label{fig: worse off}
\end{figure}
\vspace{-4mm}

Then, we fix a small amount of autonomous agents, i.e., $\alpha=0.02$, and change the travel time function of link 6. Specifically, we vary $k_a$ while keeping $b_6 = 18.3$, and vary $b_6$ while keeping $k_6 = 1$. Figure \ref{fig: changing k} and \ref{fig: changing b} show the equilibrium social cost as a function of $k_6$ and $b_6$, respectively. In these cases, the equilibrium uses the same set of active paths for both agent types. From the figures, we see that as $k_6$ and $b_6$ increase, the social cost also increases, because routing autonomous agents onto link 6 becomes more costly.

\vspace{-6mm}
\begin{figure}[ht]
    \centering
    \subfigure[Changing $k_6$]{
        \begin{minipage}[t]{0.5\linewidth}
        \centering
        \includegraphics[width=\linewidth]{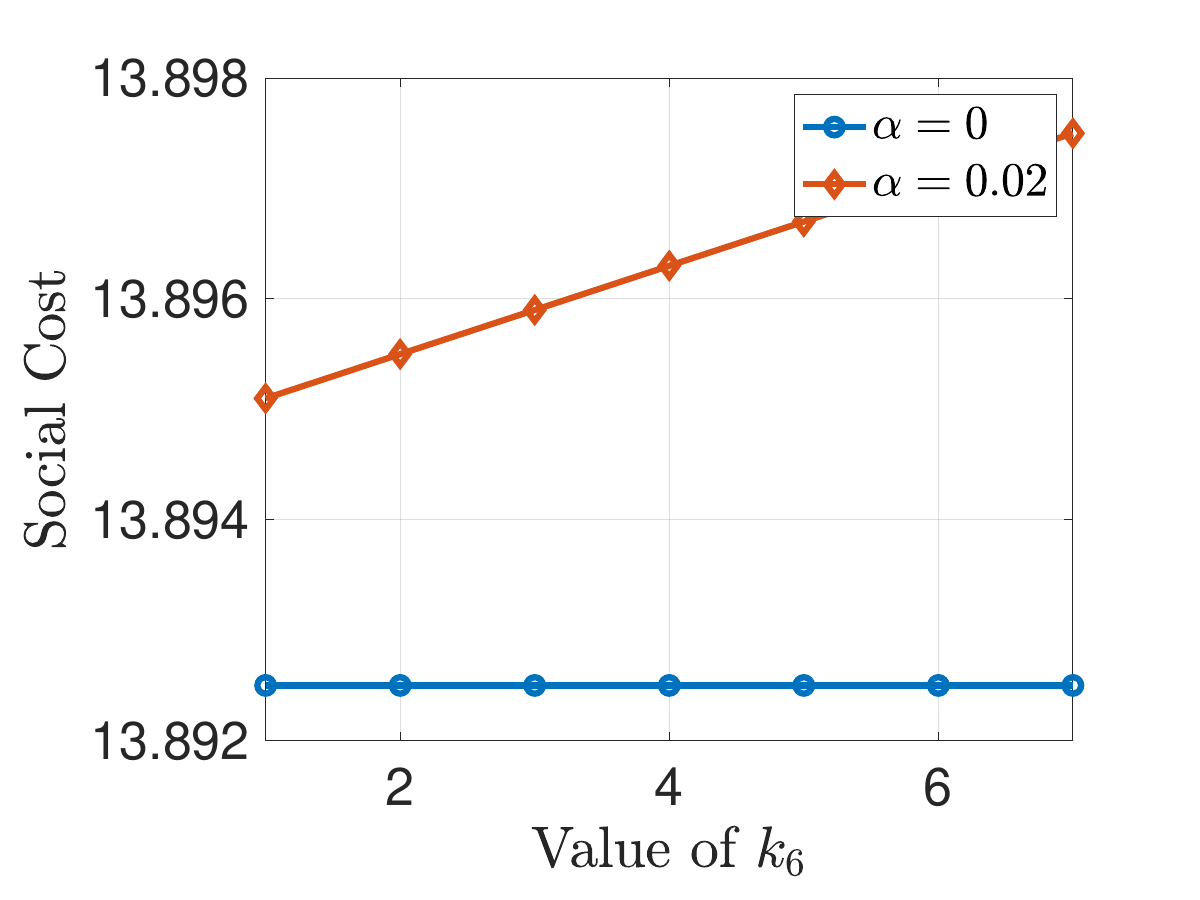}
        \label{fig: changing k}
        \vspace{-13mm}
    \end{minipage}
    }%
    \subfigure[Changing $b_6$]{
    \begin{minipage}[t]{0.5\linewidth}
        \centering
        \includegraphics[width=\linewidth]{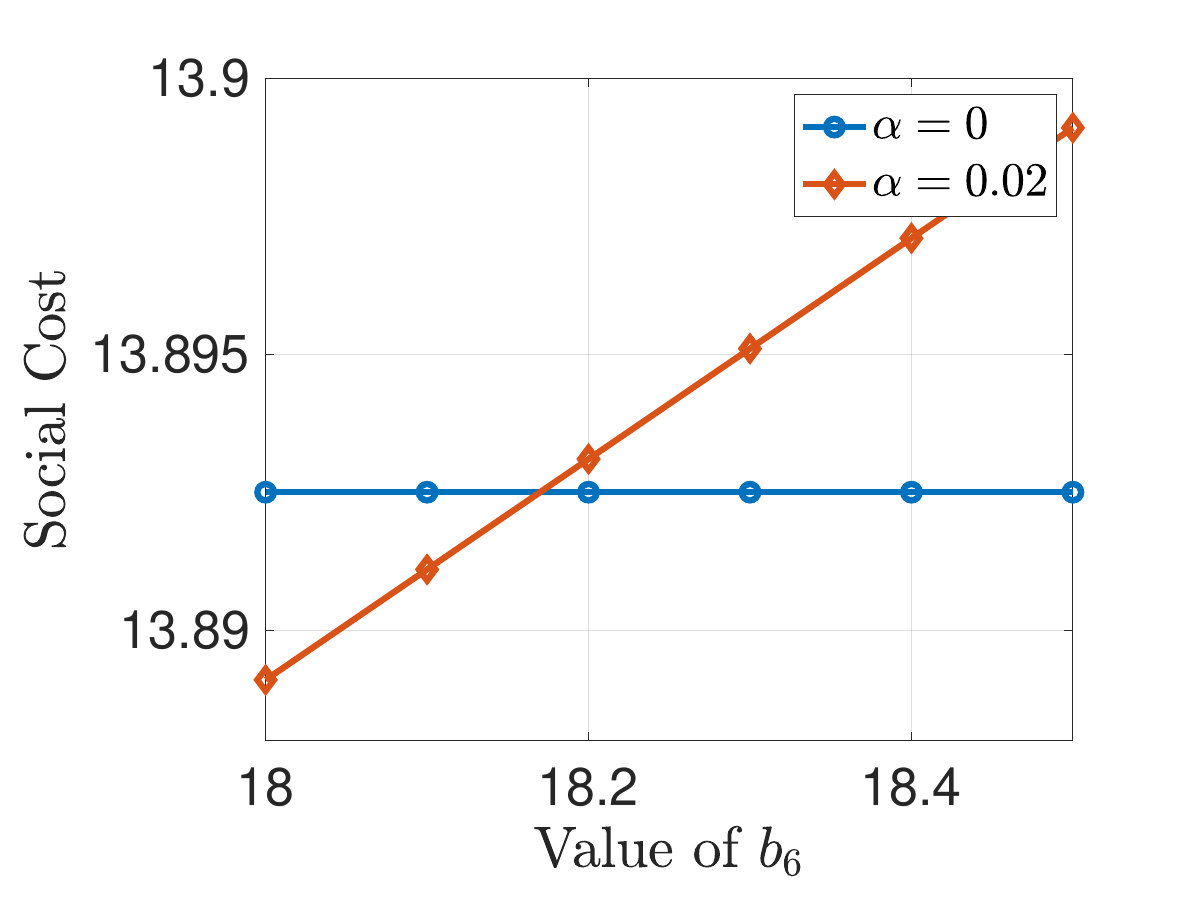}
        \label{fig: changing b}
        \vspace{-13mm}
    \end{minipage}
    }
        \vspace{-4mm}
    \caption{Social cost vs. travel time function parameters of link 6.}
\end{figure}
\vspace{-5mm}

\section{Discussion on Centralized Autonomous Routing}\label{sec: centralized}

In the previous sections, we study the decentralized setting, in which autonomous agents independently choose routes so as to minimize the social cost. A natural question is whether a centralized controller that routes the entire autonomous flow jointly can achieve a different, possibly better, equilibrium outcome.

In the centralized scenario, a central controller routes all autonomous agents to minimize the social cost, while the human agents remain decentralized and self-interested, choosing paths to minimize their own travel times. For a fixed human flow $f^H$, the centralized autonomous-routing problem is
\begin{align*}\tag{P}
    \min_{f^A} \qquad &S(f^{H*}, f^A) \\
    \text{s.t.} \qquad &\sum_{p \in \mathcal{P}} x_p^A = \alpha, \; f^A = \Delta x^A, \; x^A \geq \mathbf{0}.
\end{align*}
Thus, the overall interaction is still a simultaneous game between self-interested human agents and the central controller of autonomous agents. 

An equilibrium in the centralized setting is a feasible pair $(f^{H*},f^{A*})$ such that:

\begin{enumerate}
    \item Given autonomous agents' best response $f^{A*}$, human agents' best response $f^{H*}$ satisfies Wardrop's first principle:
    \begin{align*}
        C_p^H(f^{H*}, f^{A*}) &\begin{cases}= \lambda^H, & \text{if } x_p^{H*} > 0, \\ \geq \lambda^H, & \text{if } x_p^{H*} = 0;\end{cases}
    \end{align*}
    \item Given human agents' best response $f^{H*}$, autonomous agents' best response $f^{A*}$ solves problem (P).
\end{enumerate}

The following theorem shows that, under convex travel times, the centralized and decentralized formulations yield exactly the same equilibrium outcome.

\begin{theorem}\label{thm: centralized equals decentralized}
Assume that each link travel time function $t_a(\cdot)$ is convex, the equilibrium link flow pattern $(f^{H*}, f^{A*})$ in the decentralized scenario is also an equilibrium link flow pattern in the centralized scenario, and vice versa.
\end{theorem}
\begin{proof}
    The human agents' best response is characterized identically in both scenarios. For the autonomous agents, we show that the solution to problem (P) is equivalent to solving the VI problem in Lemma \ref{lemma: VI formulation} with respect to $f^A$, given $f^{H*}$. For details, see Appendix \ref{appendix: proof of thm 6}.
\end{proof}
Therefore, under convex link travel times, a central planner cannot obtain a better equilibrium outcome merely by coordinating the entire autonomous fleet, although it is worth noting that the result only concerns equilibrium outcomes.

\section{Conclusion}\label{sec: conclusion}

This paper studies the traffic assignment problem in a mixed-autonomy setting where self-interested human drivers coexist with altruistic autonomous agents. By formulating the equilibrium as a variational inequality, we establish the existence of equilibrium under general cost functions and the uniqueness of the aggregated link flow and social cost under BPR travel time functions. We derive sufficient conditions on the network and cost structures that determine whether including autonomous agents improves, deteriorates, or has no effect on social cost. We further show that decentralized altruistic routing and centralized fleet management yield the same equilibrium outcome, under convex travel times. As a future direction, it would be interesting to characterize how the social cost evolves over the full range of $\alpha$.

\bibliographystyle{IEEEtran}
\bibliography{refs}

\clearpage
\appendices
\section{Proof of Theorem \ref{thm: social cost to path cost}}\label{appendix: proof of thm 1}
Here, we prove the general case with multiple O/D pairs. The proof for the single O/D pair case follows directly. For this purpose, we need to first refine our notations. Let $\mathcal{J}$ denote the set of O/D pairs in the network. For each O/D pair $w \in \mathcal{J}$, let $\mathcal{P}_w$ denote the set of paths connecting the O/D pair $w$. Let $\mathcal{P} = \bigcup_{w \in \mathcal{J}} \mathcal{P}_w$ denote the set of all paths in the network. The demand for each O/D pair $w \in \mathcal{J}$ is denoted by $d_w$. The fraction of autonomous agents for O/D pair $w$ is denoted by $\alpha_w \in [0,1]$, and thus the demand of human and autonomous agents for O/D pair $w$ are $(1-\alpha_w)d_w$ and $\alpha_w d_w$, respectively. The link flows and path flows are defined the same as before, except that now we have path flows $x_p^H$ and $x_p^A$ for all $p \in \mathcal{P}_w$ and all $w \in \mathcal{J}$. The relation between link flows and path flows is still given by \eqref{eq: link-path relation}. The flow needs to meet the demand of all O/D pairs, which leads to the following constraints:
\begin{equation}\label{eq: flow meets demand multi O/D}
\sum_{p \in \mathcal{P}_w} x_p^H = (1-\alpha_w)d_w, \; \sum_{p \in \mathcal{P}_w} x_p^A = \alpha_w d_w, \quad \forall w \in \mathcal{J}.
\end{equation}
The definition of equilibrium is also refined accordingly. 
\begin{definition}[Equilibrium - Multiple O/D pairs]
A link flow pattern $(f^{H*}, f^{A*}) \in \mathcal{K}$ is an equilibrium if no user has any incentive to alter their path. Mathematically, for every O/D pair $w \in \mathcal{J}$ and every path $p \in \mathcal{P}_w$:
\begin{enumerate}
    \item If $x_p^{H*} > 0$, consider an alternative pattern $(f^{H'}, f^{A*})$ such that $x_p^{H'} = x_p^{H*}-\epsilon$ and $x_q^{H'} = x_q^{H*}+\epsilon$ with $\epsilon \to 0$ for any $q \in \mathcal{P}_w \setminus \{p\}$, and $x_r^{H'} = x_r^{H*}$ for all $r \in \mathcal{P} \setminus \{p,q\}$, there holds $C_q^H(f^{H'},f^{A*}) \ge C_p^H(f^{H*},f^{A*})$.
    \item If $x_p^{A*} > 0$, consider an alternative pattern $(f^{H*}, f^{A'})$ such that $x_p^{A'} = x_p^{A*}-\epsilon$ and $x_q^{A'} = x_q^{A*}+\epsilon$ with $\epsilon \to 0$ for any $q \in \mathcal{P}_w \setminus \{p\}$, and $x_r^{A'} = x_r^{A*}$ for all $r \in \mathcal{P} \setminus \{p,q\}$, there holds $S(f^{H*},f^{A'}) \ge S(f^{H*},f^{A*})$.
\end{enumerate}
\end{definition}

For human agents, the equilibrium condition is still characterized by Wardrop's first principle, which must hold for every O/D pair $w \in \mathcal{J}$ and every path $p \in \mathcal{P}_w$:
\begin{align*}
    C_p^H(f^{H*}, f^{A*}) &\begin{cases}= \lambda_w^H, & \text{if } x_p^{H*} > 0, \\ \geq \lambda_w^H, & \text{if } x_p^{H*} = 0,\end{cases}
\end{align*}
where $(f^{H*}, f^{A*})$ is the equilibrium link flow pattern, if it exists, and $\lambda_w^H$ is the equilibrium cost for human agents for O/D pair $w$.

Now, we focus on the second condition regarding autonomous agents. We show that it can be equivalently characterized by a path-based condition similar to Wardrop's principle.

At the equilibrium, no autonomous agent can reduce their objective, i.e.,
the social cost, by unilaterally changing their path. For any O/D pair $w \in \mathcal{J}$ and any path $p \in \mathcal{P}_w$ with $x_p^{A*} > 0$, consider an alternative pattern $(f^{H*}, f^{A'})$ such that $x_p^{A'} = x_p^{A*}-\epsilon$ and $x_q^{A'} = x_q^{A*}+\epsilon$ with $\epsilon \to 0$ for any $q \in \mathcal{P}_w \setminus \{p\}$, and $x_r^{A'} = x_r^{A*}$ for all $r \in \mathcal{P} \setminus \{p,q\}$. The corresponding link flow pattern is $(f^{H*}, f^{A'})$, where $f_a^{A'} = f_a^{A*} - \epsilon \Delta_{ap} + \epsilon \Delta_{aq}$ for all $a \in \mathcal{L}$. Denote $f_a^* = f_a^{H*}+f_a^{A*}$. The change in social cost is given by
\begin{align*}
    &S(f^{H*}, f^{A'}) - S(f^{H*}, f^{A*}) \\
    &= \sum_{a \in \mathcal{L}} \Delta_{ap}(1-\Delta_{aq}) \big[ (f_a^* - \epsilon) t_a(f_a^* - \epsilon) - f_a^* t_a(f_a^*) \big]\\
    &\quad + \sum_{a \in \mathcal{L}} \Delta_{aq}(1-\Delta_{ap}) \big[ (f_a^* + \epsilon) t_a(f_a^* + \epsilon) - f_a^* t_a(f_a^*) \big].
\end{align*}
This is because only the links on path $p$ or $q$ are affected by the change in path flow. Further, the links that are on both paths $p$ and $q$ (i.e., $\Delta_{ap} = 1$ and $\Delta_{aq} = 1$) are not affected either, since the decrease in flow on path $p$ is exactly offset by the increase in flow on path $q$.

Since the equilibrium condition requires $S(f^{H*}, f^{A'}) - S(f^{H*}, f^{A*}) \geq 0$, we have 
\begin{align*}
    &\sum_{a \in \mathcal{L}} \Delta_{aq}(1-\Delta_{ap}) \big[ (f_a^* + \epsilon) t_a(f_a^* + \epsilon) - f_a^* t_a(f_a^*) \big]\\
    \ge &\sum_{a \in \mathcal{L}} \Delta_{ap}(1-\Delta_{aq}) \big[f_a^* t_a(f_a^*) - (f_a^* - \epsilon) t_a(f_a^* - \epsilon)\big].
\end{align*}
Dividing both sides by $\epsilon$ and taking the limit as $\epsilon \to 0$, since $t_a(\cdot)$ is differentiable, we have
\begin{align*}
    &\sum_{a \in \mathcal{L}} \Delta_{aq}(1-\Delta_{ap}) \frac{d (f_a t_a(f_a))}{d f_a} \bigg|_{f_a = f_a^*} \\
    \ge &\sum_{a \in \mathcal{L}} \Delta_{ap}(1-\Delta_{aq}) \frac{d (f_a t_a(f_a))}{d f_a} \bigg|_{f_a = f_a^*}.
\end{align*}
Rearranging the terms, we have
\begin{equation*}
    \sum_{a \in \mathcal{L}} \Delta_{aq} \frac{d (f_a t_a(f_a))}{d f_a} \bigg|_{f_a = f_a^*} \ge \sum_{a \in \mathcal{L}} \Delta_{ap} \frac{d (f_a t_a(f_a))}{d f_a} \bigg|_{f_a = f_a^*},
\end{equation*}
which is equivalent to $C_q^A(f^*) \ge C_p^A(f^*)$. This needs to hold for any $q \in \mathcal{P}_w \setminus \{p\}$. Thus, we have shown that if $x_p^{A*} > 0$, then $C_p^A(f^*) \le C_q^A(f^*)$ for all $q \in \mathcal{P}_w$. The statement holds for all $w \in \mathcal{J}$. Therefore, the equilibrium condition for autonomous agents can be characterized as follows: for every O/D pair $w \in \mathcal{J}$ and every path $p \in \mathcal{P}_w$,
\begin{align*}
    C_p^A(f^{H*}, f^{A*}) &\begin{cases}= \lambda_w^A, & \text{if } x_p^{A*} > 0, \\ \geq \lambda_w^A, & \text{if } x_p^{A*} = 0,\end{cases}
\end{align*}
where $\lambda_w^A$ is the equilibrium cost for autonomous agents for O/D pair $w$. The converse follows by reversing the inequalities: if the Wardrop-like condition holds, then $C_q^A(f^*) \ge C_p^A(f^*)$ implies the directional derivative of $S$ is non-negative for any feasible deviation, so no agent can reduce the social cost. \qed

\section{Proof of Theorem \ref{thm: uniqueness of aggregated link flow}}\label{appendix: proof of thm 2}
We prove the uniqueness of the aggregated link flow by showing that the VI problem \eqref{eq: VI} has a unique solution in terms of the aggregated flow $f = f^H + f^A$. 

Assume the contrary that there exist two distinct equilibrium link flow patterns $\Bar{f}^*= \begin{bmatrix} f^{H*} \\ f^{A*} \end{bmatrix}$ and $\Bar{f}^{\dagger}= \begin{bmatrix} f^{H\dagger} \\ f^{A\dagger} \end{bmatrix}$ such that the aggregated link flows are different, i.e., $f^* = f^{H*} + f^{A*} \neq f^\dagger = f^{H\dagger} + f^{A\dagger}$ for at least one link $a \in \mathcal{L}$. By the VI formulation in Lemma \ref{lemma: VI formulation}, we have
\begin{equation}\label{eq: VI 1}
c(\Bar{f}^*) \cdot (\Bar{f} - \Bar{f}^*) \geq 0, \quad \forall (f^H, f^A) \in \mathcal{K},
\end{equation}
and
\begin{equation}\label{eq: VI 2}
c(\Bar{f}^\dagger) \cdot (\Bar{f} - \Bar{f}^\dagger) \geq 0, \quad \forall (f^H, f^A) \in \mathcal{K},
\end{equation}
where $\Bar{f} = \begin{bmatrix} f^H \\ f^A \end{bmatrix}$. Substituting $\Bar{f} = \Bar{f}^\dagger$ into \eqref{eq: VI 1} and $\Bar{f} = \Bar{f}^*$ into \eqref{eq: VI 2}, we have
\begin{equation}\label{eq: VI 3}
c(\Bar{f}^*) \cdot (\Bar{f}^\dagger - \Bar{f}^*) \geq 0,
\end{equation}
and
\begin{equation}\label{eq: VI 4}
c(\Bar{f}^\dagger) \cdot (\Bar{f}^* - \Bar{f}^\dagger) \geq 0.
\end{equation}
Adding \eqref{eq: VI 3} and \eqref{eq: VI 4}, we have
\begin{equation}\label{eq: VI 5}
\big( c(\Bar{f}^*) - c(\Bar{f}^\dagger) \big) \cdot (\Bar{f}^\dagger - \Bar{f}^*) \ge 0.
\end{equation}
Denote the left-hand side of \eqref{eq: VI 5} by $LHS$. By the definition of $c(\cdot)$, we have
\begin{align*}
    LHS &= \sum_{a \in \mathcal{L}} \big(c_a^H(f^*) - c_a^H(f^\dagger)\big) \big(f_a^{H\dagger} - f_a^{H*}\big) \\
    &\quad + \sum_{a \in \mathcal{L}} \big(c_a^A(f^*) - c_a^A(f^\dagger)\big) \big(f_a^{A\dagger} - f_a^{A*}\big).
\end{align*}
Note that $c_a^H(f) = t_a(f_a)$ and $c_a^A(f) = t_a(f_a) + f_a t_a'(f_a) = c_a^H(f) + f_a t_a'(f_a)$, where $t_a'(f_a)$ is the derivative of $t_a(f_a)$. Thus, we have
\begin{align*}
    LHS &= \sum_{a \in \mathcal{L}} \big(c_a^A(f^*) - c_a^A(f^\dagger)\big) \big(f_a^{A\dagger} - f_a^{A*}\big)\\
    &\quad + \sum_{a \in \mathcal{L}} \big(c_a^A(f^*) - c_a^A(f^\dagger)\big) \big(f_a^{H\dagger} - f_a^{H*}\big)\\
    &\quad - \sum_{a \in \mathcal{L}} \big(f_a^* t_a'(f_a^*) - f_a^\dagger t_a'(f_a^\dagger)\big) \big(f_a^{H\dagger} - f_a^{H*}\big),
\end{align*}
which can be rearranged as
\begin{align*}
    LHS &= \sum_{a \in \mathcal{L}} \big(c_a^A(f^*) - c_a^A(f^\dagger)\big) \big(f_a^{\dagger} - f_a^{*}\big)\\
    &\quad - \sum_{a \in \mathcal{L}} \big(f_a^* t_a'(f_a^*) - f_a^\dagger t_a'(f_a^\dagger)\big) \big(f_a^{H\dagger} - f_a^{H*}\big).
\end{align*}
We analyze the two terms separately. For the first term, since $c_a^A(f) = t_a(f_a) + f_a t_a'(f_a)$ and $t_a(\cdot)$ is of the form $t_a(f_a) = k_a f_a^{n} + b_a$ with $k_a > 0$, $n \ge 1$, and $b_a \ge 0$, $c_a^A(f)$ is strictly increasing in $f_a$. Thus, for any link $a$ with $f_a^* \neq f_a^\dagger$, we have $\big(c_a^A(f^*) - c_a^A(f^\dagger)\big) \big(f_a^{\dagger} - f_a^{*}\big) < 0$. For any link $a$ with $f_a^* = f_a^\dagger$, the term is zero. Therefore, the first term is strictly negative since there exists at least one link $a$ with $f_a^* \neq f_a^\dagger$ by assumption.

Next, we consider the second term. Since the point $\Bar{f} = \begin{bmatrix} f^{H\dagger} \\ f^{A*} \end{bmatrix}$ is feasible in $\mathcal{K}$, we substituting it into \eqref{eq: VI 1} and obtain
\begin{equation*}
c^H(\Bar{f}^*) \cdot (f^{H\dagger} - f^{H*}) = \sum_{a \in \mathcal{L}} t_a(f_a^*) (f_a^{H\dagger} - f_a^{H*}) \geq 0.
\end{equation*}
Similarly, substituting $\Bar{f} = \begin{bmatrix} f^{H*} \\ f^{A\dagger} \end{bmatrix}$ into \eqref{eq: VI 2}, we have
\begin{equation*}
c^H(\Bar{f}^\dagger) \cdot (f^{H*} - f^{H\dagger}) = \sum_{a \in \mathcal{L}} t_a(f_a^\dagger) (f_a^{H*} - f_a^{H\dagger}) \geq 0.
\end{equation*}
Adding the preceding two inequalities, we have
\begin{equation*}
\sum_{a \in \mathcal{L}} \big(t_a(f_a^*) - t_a(f_a^\dagger)\big) (f_a^{H\dagger} - f_a^{H*}) \geq 0.
\end{equation*}
Since $t_a(f_a) = k_a f_a^{n} + b_a$, we then have
\begin{align*}
f_a^* t_a'(f_a^*) - f_a^\dagger t_a'(f_a^\dagger) &= n k_a \big( (f_a^*)^{n} - (f_a^\dagger)^{n} \big) \\
&= n \big( t_a(f_a^*) - t_a(f_a^\dagger) \big).
\end{align*}
Therefore, we obtain
\begin{equation*}
\sum_{a \in \mathcal{L}} \big(f_a^* t_a'(f_a^*) - f_a^\dagger t_a'(f_a^\dagger)\big) (f_a^{H\dagger} - f_a^{H*}) \geq 0.
\end{equation*}
Combining the results for the two terms, we have shown that $LHS < 0$, which contradicts \eqref{eq: VI 5}. Thus, the assumption is false, and the aggregated link flow at equilibrium is unique. \qed

\section{Proof of Lemma \ref{lemma: improvement condition}}\label{appendix: proof of lemma 3}
First, we observe that the feasible set $\mathcal{K}$ can be decomposed into the Cartesian product of two sets $\mathcal{K}^H$ and $\mathcal{K}^A$, where
\begin{align*}
    \mathcal{K}^H = \big\{ f^H \;|\; \exists x^H \ge \mathbf{0}, f^H = \Delta x^H, \sum_{p \in \mathcal{P}} x_p^H = (1-\alpha) \big\},
\end{align*}
and 
\begin{align*}
    \mathcal{K}^A = \big\{ f^A \;|\; \exists x^A \ge \mathbf{0}, f^A = \Delta x^A, \sum_{p \in \mathcal{P}} x_p^A = \alpha \big\}.
\end{align*}
Based on Lemma \ref{lemma: VI formulation}, for the system with autonomous agents, substituting $\Bar{f} = \begin{bmatrix} f^{H*} \\ f^{A} \end{bmatrix}$ into the VI problem \eqref{eq: VI}, we have
\begin{equation}\label{eq: VI human}
c^A(\Bar{f}^*) \cdot (f^{A} - f^{A*}) \geq 0, \quad \forall f^A \in \mathcal{K}^A.
\end{equation}

By Assumption \ref{assump: BPR}, we have $t_a(f_a)$ being a convex and strictly increasing function. Hence, $f_a t_a(f_a)$ is strictly convex and strictly increasing. Therefore, 
\begin{equation*}
    \tilde{f}_a^* t_a(\tilde{f}_a^*) - f_a^* t_a(f_a^*) \ge \frac{d (f_a t_a(f_a))}{d f_a} \bigg|_{f_a = f_a^*} ( \tilde{f}_a^* - f_a^* ), 
\end{equation*}
with equality if and only if $\tilde{f}_a^* = f_a^*$. Summing over all links $a \in \mathcal{L}$, and by the definition of $c_a^A(f)$ and $S(f)$, we have
\begin{equation*}
    \tilde{S}^* - S^* = S(\tilde{f}^*) - S(f^{*}) \ge c^A(\Bar{f}^*) \cdot (\tilde{f}^* - f^{*}),
\end{equation*}
with equality if and only if $\tilde{f}^* = f^{*}$. Rearranging the terms, we have
\begin{equation}\label{eq: social cost and link cost}
    \tilde{S}^* - S^* \ge c^A(\Bar{f}^*) \cdot \big[ (\tilde{f}^* - f^{H*}) - f^{A*} \big].
\end{equation}
By the condition given in the lemma, there exists $\tilde{x}^*$ and $(x^{H*}, x^{A*})$ such that $x^{H*} \le \tilde{x}^{*}$. Thus, we have $\tilde{x}^{*} - x^{H*} \ge \mathbf{0}$ and $\sum_{p \in \mathcal{P}} (\tilde{x}_p^{*} - x_p^{H*}) = \alpha$. Let $f^A = \Delta (\tilde{x}^{*} - x^{H*})$, we have $f^A \in \mathcal{K}^A$. Substituting $f^A$ into \eqref{eq: VI human}, we obtain
\begin{equation*}
c^A(\Bar{f}^*) \cdot \big[ (\tilde{f}^* - f^{H*}) - f^{A*} \big] \geq 0.
\end{equation*}
Combining this with \eqref{eq: social cost and link cost}, we have $\tilde{S}^* - S^* \ge 0$. Note that $S^*$ only depends on the equilibrium aggregated link flow $f^{*}$, which is unique by Theorem \ref{thm: uniqueness of aggregated link flow}. Thus, the inequality holds for any equilibrium. Therefore, we conclude that $S^* \le \tilde{S}^*$.

Moreover, if $f^{*} \neq \tilde{f}^{*}$, the inequality in \eqref{eq: social cost and link cost} is strict. Thus, we have $\tilde{S}^* > S^*$.  \qed

\section{Improvement on Path Multigraphs}\label{appendix: proof of thm 3}

In this appendix, we constructively verify that the condition of Lemma \ref{lemma: improvement condition} ($x^{H*} \le \tilde{x}^*$) is always satisfied on path multigraphs. We begin with a \textit{simple path multigraph}, which is a network with only two nodes, the origin and destination, connected by multiple parallel links.

\begin{lemma}\label{lemma: parallel links}
In a simple path multigraph, the social cost always improves with the introduction of autonomous agents, i.e., $S^* \le \tilde{S}^*$. The inequality is strict as long as mixed autonomy equilibrium is not an equilibrium for the baseline system. 
\end{lemma}
\begin{proof}
    Since there are only two nodes connected by $L$ parallel links, there are $L$ paths in total, each corresponding to one link. For simplicity, we also call the links as paths in this proof.

    Consider the mixed autonomy system. Let $\mathcal{P}^H = \{ p \in \mathcal{P} \;|\; x_p^{H*} > 0 \}$ and $\mathcal{P}^A = \{ p \in \mathcal{P} \;|\; x_p^{A*} > 0 \}$ denote the sets of paths used by human and autonomous agents at equilibrium, respectively. By the equilibrium condition in Theorem \ref{thm: social cost to path cost}, we have
    \begin{align*}
        \begin{cases}
            C_p^H(x_p^{*}) = \lambda^H, \quad \forall p \in \mathcal{P}^H,\\
            C_p^H(x_p^{*}) \ge \lambda^H, \quad \forall p \in \mathcal{P} \setminus \mathcal{P}^H,
        \end{cases}
    \end{align*}
    where $x_p^{*} = x_p^{H*} + x_p^{A*}$. Note that each $C_p^H$ only depends on the flow on path $p$, i.e., $C_p^H(x_p^{*}) = t_p(x_p^{*})$, since there is no interaction among different paths. 

    Next, we consider an iterative process to construct an equilibrium path flow pattern for the baseline system, starting from the equilibrium in the mixed autonomy system. 

    \begin{algorithm}
    \caption{Constructing an equilibrium path flow pattern for the baseline system}
    \begin{algorithmic}[1]
        \State \textbf{Input:} Mixed autonomy system equilibrium path flow pattern $(x^{H*}, x^{A*})$, human agents' cost $\lambda^H$
        \State \textbf{Initialize:} Set $\tilde{x}_p = x_p^{*}$ for all $p \in \mathcal{P}$
        \While{$\mathcal{P}^A \neq \emptyset$}
            \State Pick a path $q \in \arg\min_{p \in \mathcal{P}^A} C_p^H(\tilde{x}_p)$:
            \If{$C_q^H(\tilde{x}_q) = \lambda^H$}  
                \State $\mathcal{P}^H \gets \mathcal{P}^H \cup \{q\}, \quad \mathcal{P}^A \gets \mathcal{P}^A \setminus \{q\}$
            \Else
                \State Find appropriate $\{\varepsilon_p\}_{p \in \mathcal{P}}$ such that $\varepsilon_q < 0$, $\varepsilon_p > 0$
                \Statex\hspace*{3em}for all $p \in \mathcal{P}^H$ and $\varepsilon_p = 0$ o.w., $\sum_{p \in \mathcal{P}} \varepsilon_p = 0$,
                \Statex\hspace*{3em}$C_p^H(\tilde{x}_p+\varepsilon_p) = \lambda'$, $\forall p \in \mathcal{P}^H$, and either of the
                \Statex\hspace*{3em}following holds:
                \Statex\hspace*{3em}1) $\varepsilon_q=-\tilde{x}_q$, $\lambda' = \min_{p \in \mathcal{P}} C_p^H(\tilde{x}_p+\varepsilon_p)$
                \Statex\hspace*{3em}2) $\varepsilon_q>-\tilde{x}_q$, $\lambda' = \min_{p \in \mathcal{P} \setminus \mathcal{P}^H} C_p^H(\tilde{x}_p+\varepsilon_p)$      
                % \State $\tilde{x}_q \gets \tilde{x}_q - \varepsilon$, $\tilde{x}_p \gets \tilde{x}_p + \varepsilon_p, \forall p \in \mathcal{P}^H$
                % \Statex\hspace*{3em}with appropriate $\varepsilon, \varepsilon_p>0$ such that $\sum_{p} \varepsilon_p = \varepsilon$,
                % \Statex\hspace*{3em}$C_p^H(\tilde{x}_p) = \lambda'$, $\forall p \in \mathcal{P}^H$, and either of the
                % \Statex\hspace*{3em}following holds:
                % \Statex\hspace*{3em}1) $\varepsilon=\tilde{x}_q$, $\lambda' = \min_{r \in \mathcal{P}} C_r^H(\tilde{x}_r)$
                % \Statex\hspace*{3em}2) $\varepsilon<\tilde{x}_q$, $\lambda' = \min_{r \in \mathcal{P} \setminus \mathcal{P}^H} C_r^H(\tilde{x}_r)$
                \If{$\varepsilon_q=-\tilde{x}_q$}
                    \State $\mathcal{P}^A \gets \mathcal{P}^A \setminus \{q\}$
                \EndIf
                % \While{$C_q^H(\tilde{x}_q) > \lambda^H$ and $\tilde{x}_q > 0$}
                %     \State $\mathcal{P}^H \gets \arg\min_{p \in \mathcal{P}} C_p^H(\tilde{x}_p)$
                %     \State $\tilde{x}_q \gets \tilde{x}_q - \varepsilon$, $\tilde{x}_p \gets \tilde{x}_p + \varepsilon_p, \forall p \in \mathcal{P}^H$
                %     \Statex\hspace*{6em}with appropriate $\varepsilon>0$, $\varepsilon_p$ such that
                %     \Statex\hspace*{6em}$\sum_{p \in \mathcal{P}^H} \varepsilon_p = \varepsilon$ and either
                %     \Statex\hspace*{6em}$
                %     \begin{cases}
                %         C_p^H(\tilde{x}_p) = \lambda', \quad\forall p \in \mathcal{P}^H,\\
                %         C_p^H(\tilde{x}_p) \ge \lambda', \quad\forall p \in \mathcal{P},\\
                %         \tilde{x}_q = 0,
                %     \end{cases}
                %     $
                %     \Statex\hspace*{6em}or
                %     \State $\lambda^H \gets \lambda'$
                %     % \State $\mathcal{P}^H \gets \{ p \in \mathcal{P} \;|\; C_p^H(\tilde{x}_p) = \lambda^H \}$
                % \EndWhile
                \State $\mathcal{P}^H \gets \mathcal{P}^H \cup \argmin_{p \in \mathcal{P} \setminus \mathcal{P}^H} C_p^H(\tilde{x}_p+\varepsilon_p)$
                \State $\lambda^H \gets \lambda'$
                \State $\tilde{x}_p \gets \tilde{x}_p + \varepsilon_p$, $\forall p \in \mathcal{P}$
            \EndIf
        \EndWhile
        \State \textbf{Output:} Equilibrium path flow pattern $\tilde{x}$ for the baseline system
    \end{algorithmic}
    \end{algorithm}

    The intuition behind the algorithm is to iteratively mark autonomous agents as human agents, while ensuring that human agents have no incentive to change their paths at each step. Specifically, we keep track of the set of paths used by human agents $\mathcal{P}^H$ and the set of paths used by autonomous agents $\mathcal{P}^A$. At each iteration, we pick a path $q$ from $\mathcal{P}^A$ with the lowest travel time cost. If its cost equals $\lambda^H$, we simply mark all agents on this path as human agents. Otherwise, we gradually decrease the flow on this path while increasing the flow on paths in $\mathcal{P}^H$ such that users on paths in $\mathcal{P}^H$ still have no incentive to deviate, and these newly added flows are marked as human agents. The amount is determined such that either path $q$ becomes unused or the travel time cost on paths in $\mathcal{P}^H$ reaches that of the next best path outside $\mathcal{P}^H$. The process continues until all autonomous agents are marked as human agents.

    We first presume step 8 is always achievable and have an overview of the algorithm. We have the following observations.
    \begin{enumerate}
        \item The algorithm will terminate in finite time, since at each iteration, either one path is removed from $\mathcal{P}^A$ (step 10) or at least one path is added to $\mathcal{P}^H$ (step 12). The maximum number of iterations is bounded by the total number of paths.
        \item Paths in $\mathcal{P}^H$ always have the same travel time cost, which is also the minimum among all paths (step 12).
        \item Only paths in $\mathcal{P}^H$ have non-zero flow at the end of the algorithm, due to the termination condition (step 3) and the update rule (step 8).
        \item During the entire process, no path will be removed from $\mathcal{P}^H$, and path flows on paths in $\mathcal{P}^H$ only increase (step 8).
    \end{enumerate}
    Therefore, the output path flow pattern $\tilde{x}$ is an equilibrium for the baseline system. Moreover, we have $\tilde{x}_p \ge x_p^{H*}$ for all $p \in \mathcal{P}$. Thus, by Lemma \ref{lemma: improvement condition}, we have $S^* \le \tilde{S}^*$. The inequality is strict as long as $x^* \neq \tilde{x}^*$.

    The rest of the proof is to show that step 8 is always achievable. First, let us remove the constraint $\sum_{p \in \mathcal{P}} \varepsilon_p = 0$ temporarily. Since $C_p^H(\cdot) = t_p(\cdot)$ is continuous, strictly increasing and unbounded from above, we can always find $\varepsilon_p > 0$ such that $t_p(\tilde{x}_p+\varepsilon_p) = \lambda'$, $\forall p \in \mathcal{P}^H$, for any given $\lambda' > \lambda^H$. In fact, the inverse function $t_p^{-1}(\cdot)$ exists, and it is also continuous and strictly increasing. Thus, we can set 
    \begin{equation*}
        \varepsilon_p = t_p^{-1}(\lambda') - \tilde{x}_p, \quad \forall p \in \mathcal{P}^H.
    \end{equation*}
    Next, we show that $\sum_{p \in \mathcal{P}} \varepsilon_p = 0$ holds for any $\varepsilon_q < 0$. To this end, we define $h(\lambda') = \sum_{p \in \mathcal{P}^H} \big(t_p^{-1}(\lambda') - \tilde{x}_p\big)$. Since each $t_p^{-1}(\cdot)$ is continuous and strictly increasing, $h(\lambda')$ is also continuous and strictly increasing. Moreover, $h(\lambda') \to \infty$ as $\lambda' \to \infty$. 

    Therefore, we have a continuous function $h(\lambda')$ that starts from $h(\lambda^H) = \sum_{p \in \mathcal{P}^H} ( \tilde{x}_p - \tilde{x}_p ) = 0$ and goes to infinity as $\lambda' \to \infty$. Thus, for any $\varepsilon_q < 0$, there exists a unique $\lambda' > \lambda^H$ such that $h(\lambda') = -\varepsilon_q$. Hence, we can always find $\{\varepsilon_p\}_{p \in \mathcal{P}}$ satisfying $\sum_{p \in \mathcal{P}} \varepsilon_p = 0$ and $C_p^H(\tilde{x}_p+\varepsilon_p) = \lambda'$, $\forall p \in \mathcal{P}^H$. The two other conditions in step 8 are to prevent over-shooting, by ensuring that either $\tilde{x}_q$ becomes zero or the travel time cost on paths in $\mathcal{P}^H$ reaches that of the next best path outside $\mathcal{P}^H$, which are achievable since $C_p^H$ is continuous and strictly increasing.

    Specifically, if (1) is satisfied, then paths in $\mathcal{P}^H$ remain the ones with the lowest cost after the update, and we will remove path $q$ from $\mathcal{P}^A$ as its flow becomes zero. If (2) is satisfied, then after the update, there exists at least one path outside $\mathcal{P}^H$ that has the same cost as paths in $\mathcal{P}^H$, and we will add these paths into $\mathcal{P}^H$. In this way, we ensure that paths in $\mathcal{P}^H$ always have the lowest cost among all paths after each update, maintaining the equilibrium condition for human agents.
\end{proof}

By Definition \ref{def: path multigraph}, any path multigraph can be decomposed into multiple simple path multigraphs with only two nodes connected by parallel links. It is not hard to see that the equilibrium in each simple path multigraph is independent of others, since there is no interaction among different simple path multigraphs. For both the baseline and mixed autonomy system, a flow pattern is an equilibrium for the entire path multigraph if and only if its corresponding flow patterns in all simple path multigraphs are equilibria for these simple path multigraphs as well. Furthermore, the social cost is additive across different simple path multigraphs. 

Therefore, by applying Lemma \ref{lemma: parallel links} to each simple path multigraph, we obtain $x^{H*} \le \tilde{x}^*$ for the entire path multigraph. By Lemma \ref{lemma: improvement condition}, it follows that $S^* \le \tilde{S}^*$, with strict inequality whenever $f^* \neq \tilde{f}^*$. \qed

\section{Proof of Theorem \ref{thm: deterioration condition}}\label{appendix: proof of thm 4}

We prove this by the following steps: first, we derive the baseline system equilibrium in an explicit form; second, we construct a feasible path flow pattern based on the equilibrium in baseline system; third, we show that the constructed flow pattern is indeed an equilibrium for the mixed autonomy system, given sufficiently small $\alpha$; finally, we compare the corresponding social cost with that of the baseline system equilibrium. 

With linear travel time assumption, we have $t_a(f_a) = k_a f_a + b_a$, where $k_a > 0$ and $b_a \ge 0$. Group $t_a(f_a)$ into a column vector function $t(f)$, where the $a$-th entry is $t_a(f_a)$, we have $t(f) = K f + b$.

Consider the baseline system. For any path $p \in \mathcal{P}$, the cost of human agents is given by 
\begin{align*}
    C_p^H(\tilde{x}^*) &= \sum_{a \in \mathcal{L}} \Delta_{ap} t_a(\tilde{f}_a^*) \\
    &= \big[ \Delta^T t(\tilde{f}^*) \big]_p \\
    &= \big[ \Delta^T (K \tilde{f}^* + b) \big]_p \\
    &= \big[ \Delta^T K \Delta \tilde{x}^* + \Delta^T b \big]_p. 
\end{align*}
Since $\tilde{x}^*$ is an equilibrium for the baseline system, $C_p^H(\tilde{x}^*) = \tilde{\lambda}^H$ holds for all paths $p$ used by human agents at equilibrium, i.e., for all $p \in \mathcal{V}$. As $\tilde{x}_p^* = 0$ for all $p \notin \mathcal{V}$, it suffices to consider the following equation:
\begin{equation}\label{eq: baseline equilibrium}
    \Delta_{\mathcal{V}}^T K \Delta_{\mathcal{V}} \tilde{x}_{\mathcal{V}}^* + \Delta_{\mathcal{V}}^T b = \tilde{\lambda}^H \mathbf{1}.
\end{equation}
Note that we use $\tilde{x}_{\mathcal{V}}^*$ to denote the sub-vector of $\tilde{x}^*$ with entries corresponding to paths in $\mathcal{V}$.

As we assume $M_{\mathcal{V}} = \Delta_{\mathcal{V}}^T K \Delta_{\mathcal{V}}$ is invertible, we can solve \eqref{eq: baseline equilibrium} to obtain
\begin{equation*}
    \tilde{x}_{\mathcal{V}}^* = M_{\mathcal{V}}^{-1} \big( \tilde{\lambda}^H \mathbf{1} - \Delta_{\mathcal{V}}^T b \big).
\end{equation*}
Since $\tilde{x}^*$ satisfies the demand constraint $\sum_{p \in \mathcal{P}} \tilde{x}_p^* = 1$, we have $\sum_{p \in \mathcal{V}} \tilde{x}_p^* = 1$. Thus, we can solve for $\tilde{\lambda}^H$ as
\begin{equation*}
    \tilde{\lambda}^H = \frac{1 + \mathbf{1}^T M_{\mathcal{V}}^{-1} \Delta_{\mathcal{V}}^T b}{\mathbf{1}^T M_{\mathcal{V}}^{-1} \mathbf{1}}.
\end{equation*}
Next, we construct a path flow pattern $(x^{H}, x^{A})$ for the mixed autonomy system as follows:
\begin{enumerate}
    \item Autonomous agents use path $q = \arg\min_{p \in \mathcal{P}} C_p^A(\tilde{x}^*)$;
    \item Human agents use and only use paths in $\mathcal{V}$, while maintaining the travel time cost on these paths equalized.
\end{enumerate}
Specifically, for autonomous agents, we let $x^A = \alpha e_q$, where $e_q$ is the unit vector with $1$ at the $q$-th entry and $0$ elsewhere. For human agents, we have the following equation:
\begin{equation*}
    \Delta_{\mathcal{V}}^T \big[ K \Delta (x^H+x^A) + b \big] = \lambda^H \mathbf{1}.
\end{equation*}
Since human agents only use paths in $\mathcal{V}$, we have $x_p^H = 0$ for all $p \notin \mathcal{V}$. Thus, it suffices to consider the following equation:
\begin{equation}\label{eq: human flow construction}
    \Delta_{\mathcal{V}}^T K \Delta_{\mathcal{V}} x_{\mathcal{V}}^{H} + \Delta_{\mathcal{V}}^T K \Delta e_q \alpha + \Delta_{\mathcal{V}}^T b = \lambda^H \mathbf{1},
\end{equation}
As $M_{\mathcal{V}} = \Delta_{\mathcal{V}}^T K \Delta_{\mathcal{V}}$ is invertible, we can solve \eqref{eq: human flow construction} to obtain
\begin{equation*}
    x_{\mathcal{V}}^{H} = M_{\mathcal{V}}^{-1} \big( \lambda^H \mathbf{1} - \Delta_{\mathcal{V}}^T K \Delta e_q \alpha - \Delta_{\mathcal{V}}^T b \big).
\end{equation*}
Since $x^{H}$ needs to satisfy the demand constraint $\sum_{p \in \mathcal{P}} x_p^{H} = 1 - \alpha$, we have $\sum_{p \in \mathcal{V}} x_p^{H} = 1 - \alpha$. Thus, we can solve for $\lambda^H$ as
\begin{align}\label{Eq: appendix temp 1}
    \lambda^H &= \frac{1 - \alpha + \mathbf{1}^T M_{\mathcal{V}}^{-1} \big( \Delta_{\mathcal{V}}^T K \Delta e_q \alpha + \Delta_{\mathcal{V}}^T b \big)}{\mathbf{1}^T M_{\mathcal{V}}^{-1} \mathbf{1}},\nonumber\\
    &= \tilde{\lambda}^H + \alpha \frac{\mathbf{1}^T M_{\mathcal{V}}^{-1} \Delta_{\mathcal{V}}^T K \Delta e_q - 1}{\mathbf{1}^T M_{\mathcal{V}}^{-1} \mathbf{1}}.
\end{align}
Let us denote $\gamma = \frac{\mathbf{1}^T M_{\mathcal{V}}^{-1} \Delta_{\mathcal{V}}^T K \Delta e_q - 1}{\mathbf{1}^T M_{\mathcal{V}}^{-1} \mathbf{1}}$. Then we can write $\lambda^H = \tilde{\lambda}^H + \alpha \gamma$. Consequently, we have
\begin{equation*}
    x_{\mathcal{V}}^{H} = \tilde{x}_{\mathcal{V}}^* + \alpha d,
\end{equation*}
where $d = \gamma M_{\mathcal{V}}^{-1} \mathbf{1} - M_{\mathcal{V}}^{-1} \Delta_{\mathcal{V}}^T K \Delta e_q$.

Next, we show that the constructed path flow pattern $(x^{H}, x^{A})$ is indeed an equilibrium for the mixed autonomy system, given sufficiently small $\alpha$. For this purpose, we need to verify the equilibrium conditions in Theorem \ref{thm: social cost to path cost}. Specifically,
\begin{enumerate}
    \item Human agents have no incentive to use paths outside $\mathcal{V}$, i.e., $C_p^H(x^{H},x^{A}) \ge \lambda^H$ for all $p \notin \mathcal{V}$;
    \item Human agents flow pattern is feasible, i.e., $x^{H} \ge \mathbf{0}$;
    \item Autonomous agents have no incentive to use paths other than $q$, i.e., $C_p^A(x^{H},x^{A}) \ge C_q^A(x^{H},x^{A})$ for all $p \in \mathcal{P}$.
\end{enumerate}

\noindent\underline{Condition (1)}: 

For any path $p \notin \mathcal{V}$, we have
\begin{align}\label{Eq: appendix temp 2}
    C_p^H(x^{H},x^{A}) &= \big[ \Delta^T (K f + b) \big]_p \nonumber\\
    &= \big[ \Delta^T K \Delta (x^{H} + x^{A}) + \Delta^T b \big]_p \nonumber\\
    &= \big[ \Delta^T K \Delta_{\mathcal{V}} x_{\mathcal{V}}^{H} + \Delta^T K \Delta e_q \alpha + \Delta^T b \big]_p \nonumber\\
    &= \big[ \Delta^T K \Delta_{\mathcal{V}} (\tilde{x}_{\mathcal{V}}^* + \alpha d) + \Delta^T K \Delta e_q \alpha + \Delta^T b \big]_p \nonumber\\
    &= C_p^H(\tilde{x}^*) + \alpha \big[ \Delta^T K \Delta_{\mathcal{V}} d + \Delta^T K \Delta e_q \big]_p.
\end{align}
Since $\lambda^H = \tilde{\lambda}^H + \alpha \gamma$, we have
\begin{align*}
    C_p^H(x^{H},x^{A}) - \lambda^H &= C_p^H(\tilde{x}^*) - \tilde{\lambda}^H \\
    &\quad + \alpha \big( \big[ \Delta^T K \Delta_{\mathcal{V}} d + \Delta^T K \Delta e_q \big]_p - \gamma \big).
\end{align*}
By the equilibrium condition in the baseline system, we have $C_p^H(\tilde{x}^*) - \tilde{\lambda}^H > 0$ for all $p \notin \mathcal{V}$. Therefore, as long as we choose $\alpha$ sufficiently small, we have $C_p^H(x^{H},x^{A}) - \lambda^H \ge 0$ for all $p \notin \mathcal{V}$. 

\noindent\underline{Condition (2)}:

For any path $p \in \mathcal{V}$, we have
\begin{equation*}
    x_p^{H} = \tilde{x}_p^* + \alpha d_p.
\end{equation*}
Since $\tilde{x}_p^* > 0$ for all $p \in \mathcal{V}$, we can choose $\alpha$ sufficiently small such that $x_p^{H} \ge 0$ for all $p \in \mathcal{V}$. For any path $p \notin \mathcal{V}$, we have $x_p^{H} = 0$. Therefore, $x^{H} \ge \mathbf{0}$.

\noindent\underline{Condition (3)}:
For any path $p \in \mathcal{P}$, we have
\begin{align*}
    C_p^A(x^{H},x^{A}) &= \big[ \Delta^T ( 2 K f + b ) \big]_p \\
    &= \big[ \Delta^T ( 2 K \Delta (x^{H} + x^{A}) + b ) \big]_p \\
    &= \big[ \Delta^T ( 2 K \Delta_{\mathcal{V}} x_{\mathcal{V}}^{H} + 2 K \Delta e_q \alpha + b ) \big]_p \\
    &= C_p^A(\tilde{x}^*) + 2 \alpha \big[ \Delta^T K \Delta_{\mathcal{V}} d + \Delta^T K \Delta e_q \big]_p.
\end{align*}
Since $q = \arg\min_{p \in \mathcal{P}} C_p^A(\tilde{x}^*)$, we have $C_p^A(\tilde{x}^*) - C_q^A(\tilde{x}^*) > 0$ for all $p \in \mathcal{P}$ and $p \neq q$. Therefore, as long as we choose $\alpha$ sufficiently small, we have $C_p^A(x^{H},x^{A}) - C_q^A(x^{H},x^{A}) \ge 0$ for all $p \in \mathcal{P}$.

Thus, we have shown that the constructed path flow pattern $(x^{H}, x^{A})$ is indeed an equilibrium for the mixed autonomy system, given sufficiently small $\alpha$. Next, we compare the corresponding social cost with that of the baseline system equilibrium. Specifically, the social cost at baseline system equilibrium is $\tilde{S}^* = \tilde{\lambda}^H$, while the social cost at mixed-autonomy system equilibrium is 
\begin{align*}
    S^* =& (1-\alpha) \lambda^H + \alpha C_q^H(x^H, x^A)\\
    =& (1-\alpha)(\tilde{\lambda}^H + \alpha \gamma) \\
    &+ \alpha\Big\{ C_q^H(\tilde{x}^*) + \alpha \big[ \Delta^T K \Delta_{\mathcal{V}} d + \Delta^T K \Delta e_q \big]_q \Big\},
\end{align*}
where the first equality is by definition of the social cost, and the second equality is based on \eqref{Eq: appendix temp 1} and \eqref{Eq: appendix temp 2}. Therefore, having $S^* > \tilde{S}^*$ is equivalent to 
\begin{align}
    C_q^H(\tilde{x}^*) - \tilde{\lambda}^H + \gamma + \alpha \Big\{\big[ \Delta^T K \Delta_{\mathcal{V}} d \!+\! \Delta^T K \Delta e_q \big]_q \!-\! \gamma \Big\} > 0. \label{eq: deterioration condition 1}
\end{align}
As $\alpha \to 0$, the left-hand side approaches $C_q^H(\tilde{x}^*) - \tilde{\lambda}^H + \gamma$. Hence, if we have
\begin{equation*}
    C_q^H(\tilde{x}^*) - \tilde{\lambda}^H + \gamma > 0,
\end{equation*} 
then $S^* > \tilde{S}^*$.
\qed

\section{Proof of Theorem \ref{thm: no effect condition}}\label{appendix: proof of thm 5}

Consider the mixed-autonomy system. Under Assumption \ref{assump: BPR}, the path cost for human agents is
\begin{equation*}
    C_p^H(f) = \sum_{a \in \mathcal{L}} \Delta_{ap} (k_a f_a^{n} + b_a).
\end{equation*}
Since all paths have the same free-flow travel time $b_0 = \sum_{a \in \mathcal{L}} \Delta_{ap} b_a$, we have
\begin{equation*}
    C_p^H(f) = \sum_{a \in \mathcal{L}} \Delta_{ap} k_a f_a^{n} + b_0, \quad \forall p \in \mathcal{P}.
\end{equation*}
Similarly, the path cost for autonomous agents is
\begin{align*}
    C_p^A(f) &= \sum_{a \in \mathcal{L}} \Delta_{ap} \big( (n+1) k_a f_a^{n} + b_a \big) \\
    &= (n+1)\sum_{a \in \mathcal{L}} \Delta_{ap} k_a f_a^{n} + b_0, \quad \forall p \in \mathcal{P}.
\end{align*}
Since $n+1 > 0$ is a positive multiplicative constant and $b_0$ is a common additive constant, both $C_p^H$ and $C_p^A$ induce the same ranking over paths: for any two paths $p, q \in \mathcal{P}$, $C_p^H(f) \le C_q^H(f)$ if and only if $C_p^A(f) \le C_q^A(f)$. Therefore, at equilibrium, both agent types rank paths identically. The mixed equilibrium conditions reduce to a single-class Wardrop problem with total demand 1 — the same as the baseline. By uniqueness of aggregated flow established in Theorem \ref{thm: uniqueness of aggregated link flow}, $f^* = \tilde{f}^*$. Consequently, $S^* = \tilde{S}^*$. \qed

\section{Proof of Theorem \ref{thm: centralized equals decentralized}}\label{appendix: proof of thm 6}
First, we decompose the feasible set $\mathcal{K}$ into the Cartesian product of two sets $\mathcal{K}^H$ and $\mathcal{K}^A$, where
\begin{align*}
    \mathcal{K}^H = \big\{ f^H \;|\; \exists x^H \ge \mathbf{0}, f^H = \Delta x^H, \sum_{p \in \mathcal{P}} x_p^H = (1-\alpha) \big\},
\end{align*}
and 
\begin{align*}
    \mathcal{K}^A = \big\{ f^A \;|\; \exists x^A \ge \mathbf{0}, f^A = \Delta x^A, \sum_{p \in \mathcal{P}} x_p^A = \alpha \big\}.
\end{align*}

Given autonomous agents' response $f^{A*}$, we observe that the human agents' best response $f^{H*}$ has the same condition in both the decentralized and centralized scenarios, which is equivalently to solving the VI problem:
\begin{equation*}
c^H(f^{H*}, f^{A*}) \cdot (f^{H} - f^{H*}) \geq 0, \quad \forall f^H \in \mathcal{K}^H.
\end{equation*}

Given human agents' response $f^{H*}$, the autonomous agents' best response $f^{A*}$ is the solution to the optimization problem:
\begin{align*}
    \min_{f^A} \quad &S(f^{H*}, f^A) \\
    \text{s.t.} \quad &f^A \in \mathcal{K}^A.
\end{align*}
Since $t_a$ is continuously differentiable, $S(f^{H*}, f^A)$ is also continuously differentiable. Plus, the feasible set $\mathcal{K}^A$ is convex and compact. Thus, by Proposition 1.2 in literature \cite{nagurney_2013_network}, $f^{A*}$ is a solution of the following VI problem:
\begin{equation*}
c^A(f^{H*}, f^{A*}) \cdot (f^{A} - f^{A*}) \geq 0, \quad \forall f^A \in \mathcal{K}^A,
\end{equation*}
Furthermore, since $t_a$ is convex and non-decreasing, $S(f^{H*}, f^A)$ is also convex in $f^A$. Thus, by Proposition 1.3 in \cite{nagurney_2013_network}, the solution to the preceding VI problem is also a solution to the optimization problem. 

The two VI problems for human and autonomous agents are equivalent to the VI formulation in Lemma \ref{lemma: VI formulation}. Therefore, the equilibrium link flow pattern $(f^{H*}, f^{A*})$ in the decentralized scenario is also an equilibrium link flow pattern in the centralized scenario, and vice versa. \qed

\end{document}